%% file: paper.tex
\documentclass[prodmode,permissions]{acmsmall-ec16}
\pdfoutput=1
\usepackage{amsmath,amsfonts,amssymb,bbm,graphics} 
\usepackage[numbers]{natbib} % for citet
\usepackage{hyperref}
\usepackage[ruled]{algorithm2e}
\renewcommand{\algorithmcfname}{ALGORITHM}
\usepackage{color}

\SetArgSty{textrm}  % for algorithm2e
\SetAlFnt{\small}
\SetAlCapFnt{\small}
\SetAlCapNameFnt{\small}
\SetAlCapHSkip{0pt}
\IncMargin{-\parindent}

\doi{XXXXXXX.XXXXXXX}% TeXSupport

\begin{document}

% Page heads
%\markboth{G. Zhou et al.}{A Multifrequency MAC Specially Designed for WSN Applications}

% Title portion
\title{Approximately Envy-Free Spectrum Allocation with Complementarities} 
\author{Dengwang Tang
\affil{University of Michigan}
Vijay Subramanian
\affil{University of Michigan}
}
% NOTE! Affiliations placed here should be for the institution where the
%       BULK of the research was done. If the author has gone to a new
%       institution, before publication, the (above) affiliation should NOT be changed.
%       The authors 'current' address may be given in the "Author's addresses:" block (below).
%       So for example, Mr. Abdelzaher, the bulk of the research was done at UIUC, and he is
%       currently affiliated with NASA.

\input{abstract.tex}

\input{keywords.tex}

\maketitle

\input{intro.tex}
\input{mechanism_setup.tex}

\input{analysis.tex}

\input{implementation_conclusion.tex}

% Bibliography
%\bibliographystyle{ACM-Reference-Format-Journals}
%\bibliography{paper-bibfile}

\input{paper.bbl}

\newpage
% Appendix
\appendix
\section*{APPENDIX}
\setcounter{section}{0}

\input{appendixA.tex}

\input{appendix_proof_envyfreeness.tex}
\input{appendixB.tex}

\input{appendixC.tex}

\input{appendixD.tex}

\input{appendixE.tex}

\end{document}

%% file: abstract.tex
\begin{abstract}
%Before Abstract is written I first put some place holder here.
%\lipsum[7]
With spectrum auctions as our prime motivation, in this paper we analyze combinatorial auctions where agents' valuations exhibit complementarities. Assuming that the agents only value bundles of size at most $k$ and also assuming that we can assess prices, we present a mechanism that is efficient, approximately envy-free, asymptotically strategy-proof and that has polynomial-time complexity. Modifying an iterative rounding procedure from assignment problems, we use the primal and dual optimal solutions to the linear programming relaxation of the auction problem to construct a lottery for the allocations and to assess the prices to bundles. The allocations in the lottery over-allocate goods by at most $k-1$ units, and the dual prices are shown to be (approximately) envy-free irrespective of the allocation chosen. We conclude with a detailed numerical investigation of a specific spectrum allocation problem. 
\end{abstract}

%% file: keywords.tex
\begin{CCSXML}
<ccs2012>
<concept>
<concept_id>10003752.10010070.10010099.10010101</concept_id>
<concept_desc>Theory of computation~Algorithmic mechanism design</concept_desc>
<concept_significance>500</concept_significance>
</concept>
<concept>
<concept_id>10003752.10010070.10010099.10010107</concept_id>
<concept_desc>Theory of computation~Computational pricing and auctions</concept_desc>
<concept_significance>500</concept_significance>
</concept>
<concept>
<concept_id>10003752.10003809.10003636.10003813</concept_id>
<concept_desc>Theory of computation~Rounding techniques</concept_desc>
<concept_significance>100</concept_significance>
</concept>
</ccs2012>
\end{CCSXML}


%\terms{Design, Algorithms, Performance}

%\keywords{Wireless sensor networks, media access control, multi-channel, radio interference, time synchronization}


%Combinatorial Auction, Supporting prices, 
%Walrasian Price, Envy-free Pricing Mechanism,

\begin{bottomstuff}
This work is supported by the National Science Foundation under grant
AST-1343381 and the SURE (Summer Undergraduate Research Experience) program at the University of Michigan. The authors would like to thank Rakesh Vohra and Thanh Nguyen for comments on an earlier draft. %See the
%Acknowledgements section before REFERENCES.

Author''s addresses: D. Tang and V. G. Subramanian, Electrical Engineering and Computer Science Department, University of Michigan; email: \url{{dwtang,vgsubram}@umich.edu}.
\end{bottomstuff}

%% file: intro.tex
\section{Introduction}

Market design is widely applied to many real-world problems that have interesting underlying resource allocation questions \cite{nisan2007algorithmic,krishna2009auction,easley2010networks}. Most of the problems involve allocation of indivisible goods where the preference exhibit complementarities with preferences over bundles of goods, and also externalities. Some of these problems, such as matching of residents to hospitals, matching of students to schools, or kidney exchanges, either explicitly bar the use of monetary transfers or prices to facilitate market-making, or cannot use prices owing to non-numeraire preferences. In many others problems, such as online sponsored search auctions, market clearing in electricity markets, or spectrum auctions, bids and prices obtained via combinatorial auctions \cite{cramton2006combinatorial,blumrosen2007combinatorial,krishna2009auction} are the mainstay of the underlying market-making. Despite the wide-applicability of combinatorial auctions in the latter class of problems, it is well-understood that without any restrictions on the agent utilities, the problems are computationally intractable. In this paper, we focus on the generic combinatorial auction problem, and under specific restrictions on the agent utilities present a randomized mechanism with polynomial-time complexity that ensures \emph{ex-post} approximate envy-freeness and asymptotic strategy-proofness. 
%~~~~The problem of resource allocation exists in different fields, such as allocating spectrum bands to companies, allocating couples to hospital residency positions, and allocating courses to students. In this paper, we focus on the problem of assigning bundles of indivisible goods to agents ensuring \emph{ex post} approximate envy-freeness and asymptotic strategy-proofness.

Our main motivation for this work is spectrum auctions and markets \cite{berry2010spectrum,milgrom1998game,bulow2009winning,cramton1997fcc,cramton2002spectrum}. Owing to interference considerations, it is easily seen that agents (service providers) obtain higher total utilities for allocations when bands are adjacent either in frequency or space as opposed to when the bands are separately allocated. Given the increased demand and utilization of the airwaves, many governments have successfully conducted auctions to license spectrum bands for use by commercial service providers, and sometimes via open-access. As many more auctions are expected in the future, understanding classes of utilities for which there exist optimal or almost optimal mechanisms with polynomial-time complexity is an important area of research.
%For the problem of spectrum band allocation, the agents are interested in getting bundles of bands rather than thinking of every band independently. For these agents, the utility of several adjacent band is usually larger than the sum of utility of the individual band. That is to say, the utility function of bands for an agent is non-linear. Therefore, we would consider a bundle as a whole in this problem.

One of the principal reasons for the intractability of the combinatorial auction problem is that it includes the knapsack problem with the additional complexity of having an exponential (in the number of goods being auctioned) number of integer variables. In addition, even the linear programming relaxation is a hard problem as the number of variables is still exponential in the number of goods. There is considerable research on polynomial-time approximation algorithms in this context \cite{nisan2000bidding,zurel2001efficient,bartal2003incentive,blumrosen2007combinatorial,mu2008truthful,dobzinski2012truthful,dobzinski2015multi,nisan2007computationally}. Starting with the assumption of a single-minded\footnote{Different agents are identified via specific bundles, say $S_i$ for agent $i$, and they have a positive constant utility only for all bundles that contain $S_i$.} buyer, the authors in \cite{blumrosen2007combinatorial} present a greedy, constant-factor, polynomial-time, and strategy-proof approximation mechanism for this problem that solicits bids and determines prices for the agents who get allocated bundles. This approximation mechanism is then generalized to larger class of utilities that can be obtained from the single-minded buyer setting using elementary operations. While this scheme provides a mechanism with many good properties, it is not guaranteed to be efficient. With the same family of utilities, if the linear-programming relaxation yields an integer solution \cite{blumrosen2007combinatorial}, then the Second Welfare Theorem insists that the dual variables can be used to determine Walrasian market clearing prices \cite{babaioff2014efficiency,blumrosen2007combinatorial,vohra2012principles} to be the assessed to agents that get allocated bundles; the linear programming relaxation can be solved in polynomial-time with the family of utilities considered. Complementary slackness also obtains \emph{envy-freeness} \cite{blumrosen2007combinatorial,vohra2012principles}, wherein no agent gets a higher return for the allocation of any other agent, and hence doesn't envy it. Furthermore, it is easily shown that in the presence of many agents, no agent gains much by being untruthful about their valuations, i.e., asymptotic \emph{strategy-proofness} also obtains. These results will be important precedents that will be one part of the related work.

Some recent developments for matching and assignment problems with complementarities, where  either the valuations are non-numeraire or where prices cannot be assessed, are also important precedents for our work. In \cite{nguyen2015assignment} and \cite{nguyen2014near} the authors consider one-sided and two-sided matchings with complementarities. An important restrictions on the utilities that they impose is to assert that agents do not value bundles of size (number of good in the bundle include multiplicities) greater than $k$ (a parameter); just as in the single-minded buyer setting, the valuation of the larger bundles can equivalently be set to the maximum of the bundles contained within. The mechanisms developed then solve the linear programming relaxation with envy-freeness explicitly accounted for as a constraint (because prices cannot assessed in such problems). The key innovation is to then present a polynomial-time integer-rounding-based lottery procedure such that the linear programming optimal solution is in the convex-hull of the integer solutions with the added property that none of the integer solutions exceed the supply constraints by more than $k-1$ units. Note that efficiency is guaranteed as the expected utility is exactly that obtained from the solution of the linear programming relaxation.

In this paper we adopt the $k$-sized bundles restrictions on the utilities from \cite{nguyen2015assignment} and \cite{nguyen2014near}, and ask whether there exists a polynomial-time efficient randomized mechanism when the valuations are numeraire, and agents have quasilinear utilities. The key difference in our problem is the ability to charge prices, and so we further look for a mechanism that \emph{a'la} Walrasian prices naturally obtain envy-freeness, instead of imposing it as a constraint in the linear programming relaxation. As mentioned earlier, spectrum auctions are one of the main motivations for us to study such mechanisms, and there are policy guidelines that are being discussed for the upcoming incentive auctions by the FCC where certain players like AT\&T and Verizon will have restrictions on the bands that they can bid on \cite{cramton2007700,shapiro2014economic}. Our $k$-bundle constraint on the utilities is a natural form of such restrictions.

%Another feature of the problem is that we need the allocation to be relatively fair for the agents. To qualitatively define "fairness", we use the concept of \emph{envy-freeness} \cite{nguyen2015assignment}, which means that any agent would prefer her bundle to other agents' bundles under the given price of goods. Under this criteria, we can say that the allocation is relatively fair.

%We also require the allocation to be relatively \emph{efficient}, which means that the social welfare (the sum of utility of the agents) under the final allocation would be maximized. However, finding an exactly optimal allocation is one kind of knapsack problem, which is NP hard. Therefore we first relax the constraints to make the problem a linear programming (LP) problem. After solving the LP, we may get a fractional solution. We then perform the lottery construction process described by Nguyen et al \cite{nguyen2015assignment}, and then get the allocation scheme which in average gives results close to the LP optimal solution.

%For the allocation we get, we will also find a supporting price associated with the allocation. Under the supporting price, all the goods would be cleared and every agent would get a bundle with approximately the best payoff.

Our mechanism for allocation, called the POPT (Priced OPT) mechanism, starts with the linear programming relaxation of the allocation problem with just the demand and supply constraints; we call this problem by LIP. We solve LIP via the simplex method to obtain an extreme point optimal and the corresponding Lagrange multipliers (dual optimal). We then modify the integer-rounding procedure in \cite{nguyen2015assignment} such that each of the integer solutions is the optimal solution of a related linear programming problem such that the dual optima include the dual optima for LIP. This is an important modification that allows us to construct the prices for our mechanism. We believe that this idea can be used in contexts. We then follow the lottery construction procedure from \cite{nguyen2015assignment} with a few small modifications. We prove the (approximate) envy-freeness of our mechanism first by showing a market-clearing property of our prices that we call as supporting the allocation, and then by using complementary slackness we demonstrate envy-freeness. 

The paper is organized as follows. In Section~\ref{sec:mechsetup} we describe the POPT mechanism and prove many properties of it, including approximate envy-freeness. We analyze the performance of the mechanism in Section~\ref{sec:analysis}. We then briefly describe the open-source implementation of the mechanism in Section~\ref{sec:implement} and conclude in Section~\ref{sec:conclusion}.

%Our mechanism for allocation, called POPT (Priced OPT), is constructed based on OPT Mechanism described in \cite{nguyen2015assignment}. OPT Mechanism is also an approximately efficient mechanism for bundle allocation problems. However, there are some important difference between the two mechanisms. While the OPT mechanism adds envy-freeness as an explicit constraint in optimization, POPT introduce prices for goods so that approximate envy-freeness can be achieved without adding an explicit constraint. Another difference is that, we modify the iterative rounding procedure a little bit so that we can make sure that any allocation that may be given by POPT is an approximately \emph{envy-free} solution.

% Head 1

%% file: mechanism_setup.tex
\section{Mechanism Setup}\label{sec:mechsetup}
We start by describing the mechanism. Thereafter, we elaborate on each of the steps of the mechanism and prove properties of it.

For the POPT mechanism, we will use the following procedure to get an approximate efficient allocation of the goods. \\

\noindent\textbf{\underline{POPT mechanism:}}
\begin{enumerate}
	\item Set up initial linear programming problem (LIP).
	\item Solve (LIP) and get solution $x^*$.
	\item Perform Lottery Construction process on $x^*$ to get integral solutions
	\item Construct a lottery of the integral solutions that (approximately) has the expected solution being $x^*$.
	\item Solve the dual of (LIP) to get POPT prices that support the allocation.
\end{enumerate}

% Head 2
\subsection{Initial Linear Programming Problem}

We assume that we have a set $N$ of agents and set $G$ of good types. For every type $j$ of good, we have the supply for the good to be $s_j$, where all goods of a given type are identical. A bundle is denoted by a vector $B\in\mathbb{Z}_+^{|G|}$ where the j-th coordinate $B_j$ denotes the number of good of type $j$ in bundle $B$. If $B_j^1\leq B_j^2$ for all $j\in G$ and $B^1\neq B^2$, we say that $B^1\prec B^2$. Denote the set of all available bundles as $\mathcal{B}$. Define the function $x_i:\mathcal{B}\mapsto \{0,1\},i\in N$ as the indicator variable for the allocation of bundle $B$ to agent $i$. Therefore, $x_i(B)=1$ means that bundle $B$ is allocated to agent $i$. Function $u_i:\mathcal{B}\mapsto \mathbb{R},i\in N$ determines  the valuations of each agent for all possible bundles, i.e. $u_i(B)=p$ means that bundle $B$ has value $p$ for agent $i$. 

To make this problem solvable in polynomial time, we assume that the agents are only interested in bundles with size less than or equal to $k$. In other words, the valuation can only be positive for bundles $B$ such that $$\sum_{j\in G}B_j\leq k~~~~~~ \forall B\in\mathcal{B}.$$ Therefore, for bundles $B$ with size larger than $k$ we fix $x_i(B)=0$; alternatively we can follow the convention for single-mind buyers and set the valuations of a bundle with size greater than $k$ to be the maximum of the valuation of its subsets. In either case, then we can  reduce the set $\mathcal{B}$ to contain only the \emph{k-bundles}, which denotes the bundles with size less than or equal to $k$. %Hence we have reduced the number of variables for the problem.
Furthermore, we relax $x_i(B)$ to take values in $[0,1]$. This way the problem of solving $x_i(B)$ can be formulated as a linear programming problem on a convex set.

The demand constraints for the allocation insist that every agent gets at most one bundle. These are given by
\begin{align*}
\sum_{B\in\mathcal{B}}x_i(B)\leq 1~~~~~~\forall i\in N.\tag{Demand}
\end{align*}
% Future improvement: more than one bundle

The supply constraints ensure that the goods are not over-allocated. They are given by
\begin{align*}
\sum_{i\in N}\sum_{B\in \mathcal{B}}B_jx_i(B)\leq s_j~~~~~~\forall j\in G.\tag{Supply}
\end{align*}

The objective function to be maximized is the total utility of the agents. Therefore the initial linear programming problem can be formulated as
\begin{align*}
\begin{array}{ccll}
&\displaystyle\max_{x\geq 0}&\displaystyle\sum_{i\in N}\sum_{B\in\mathcal{B}}u_i(B)x_i(B)&\\
&\mathrm{s.t.}&\displaystyle\sum_{B\in\mathcal{B}}x_i(B)\leq 1&\forall i\in N,\\
&&\displaystyle\sum_{i\in N}\sum_{B\in \mathcal{B}}B_jx_i(B)\leq s_j&\forall j\in G.
\end{array}
\end{align*}

We make some modifications to the coefficients of the problem, and we get
\begin{align}
\begin{array}{ccll}
&\displaystyle\max_{x\geq 0}&\displaystyle\sum_{i\in N}\sum_{B\in\mathcal{B}}w_i(B) u_i(B)x_i(B)&  \\
&\mathrm{s.t.}&\displaystyle\sum_{B\in\mathcal{B}}x_i(B)\leq 1&\forall i\in N,\\
&&\displaystyle\sum_{i\in N}\sum_{B\in \mathcal{B}}B_jx_i(B)\leq \tilde{s}_j:=s_j-\epsilon_j&\forall j\in G.
\end{array}
\tag{LIP} 
\end{align}
where $w_i(B)$ are weights which typically takes values near 1 and $\epsilon_j\ll 1$ are modification variables. For our mechanism, we choose $w_i$ and $\epsilon_j$ randomly, where $w_i(B)$ are drawn i.i.d. uniformly on $[1-\delta_w, 1+\delta_w]$, and $\epsilon_j$ are drawn i.i.d. uniformly on $[\delta_\epsilon, 2\delta_\epsilon]$. Here $\delta_w\ll 1$ and $\delta_\epsilon\ll 1$ are preset values. The reason for this is to ensure \emph{asymptotic strategy-proofness}, and this will be discussed in the Section~\ref{sec:analysis}. 

To solve the problem, we first solve (LIP) and get a solution $x^*$, which is most likely to be fractional. If it is fractional, we run the lottery construction process using the solution $x^*$ to get a few integral solutions to form a lottery to determine the resulting allocation.

% Head 3
\subsection{Obtaining Integral Solutions}

For the lottery construction we need a subroutine called Iterative Rounding (IR). Given any reward vector $c$, the Iterative Rounding procedure basically takes any point $z$ which satisfies (Demand) and (Supply) as input, and outputs an integral $\overline{x}$. This procedure can be denoted as a function $\overline{x}=\mathrm{IR}(z)$ or $\overline{x}=\mathrm{IR}(z;c)$.

\subsubsection{Iterative Rounding}
\label{sec: itr}

For the Iterative Rounding procedure, we basically follow the procedure described in \cite{nguyen2015assignment} but we also make some modifications. For any reward vector $c\in \mathbb{R}^{|N|\times|\mathcal{B}|}$ which has the same size as the input $z$, the procedure 
%is as follows
is as described in Algorithm \ref{alg:ir}.

\begin{algorithm}[!ht]
\SetAlgoNoLine
\KwIn{A point $z$ which satisfies (Demand) and (Supply)}
\KwOut{An integral point $z^{(\tau)}$ which satisfies (Demand) and (Supply+k-1)}
$G^{(0)}=G,~\mathcal{B}_i^{(0)}=\mathcal{B}~~~\forall i\in N,~\tilde{s}_j^{(0)}=\tilde{s}_j~~~\forall j\in G,~ z^{(\tau)}=z, \tau=0$\;

$N^{(\tau)}=\{i\in N:\sum_{i\in N}z_i(B)=1\}$\;

\Repeat{Forever}{
		\If{$z^{(\tau)}$ is integral}{\textbf{break}}
		\If{some but not all of $z^{(\tau)}_i(B)$ are integral}{
			
			$\mathcal{B}_i^{(\tau+1)} = \{B\in\mathcal{B}:0<z_i^{(\tau)}(B)< 1\}$\;
			
			$\tilde{s}_j^{(\tau+1)}=s_j'^{(\tau)}-\sum_{i\in N}\sum_{B\in \mathcal{B}_i^{(\tau)}\backslash\mathcal{B}_i^{(\tau+1)}}B_j z_i^{(\tau)}(B)~~~~~~\forall j\in G^{(\tau)}.$\;
		
		The updated linear programming problem is then 
		\begin{align}
		\begin{array}{ccll}
		&\displaystyle\max_{x\geq 0}&\displaystyle\sum_{i\in N}\sum_{B\in\mathcal{B}^{(\tau+1)}_i}c_i(B)x_i(B)&\\
		&\mathrm{s.t.}&x_i(B)=z_i^{(\tau)}(B)&\forall i\in N,B\in\mathcal{B}\backslash\mathcal{B}_i^{(\tau+1)}\\
		&&\displaystyle\sum_{B\in\mathcal{B}^{(\tau+1)}_i}x_i(B)= 1&\forall i\in N^{(\tau)},\\
		&&\displaystyle\sum_{B\in\mathcal{B}^{(\tau+1)}_i}x_i(B)\leq 1&\forall i\in N\backslash N^{(\tau)},\\
		&&\displaystyle\sum_{i\in N}\sum_{B\in \mathcal{B}^{(\tau+1)}_i}B_jx_i(B)\leq \tilde{s}_j^{(\tau+1)}&\forall j\in G^{(\tau)}.
		\end{array}
		\tag{ULIP}
		\end{align}
		
		Solve (ULIP) to get an \emph{extreme point} solution $z^*$, set $z^{(\tau+1)}=z^*$\;}
		\If{none of $z_i^{(\tau+1)}(B)$ is integral}{
		$G^{(\tau+1)}=\{j\in G^{(\tau)}: \sum_{i\in N}\sum_{B\in \mathcal{B}_i^{(\tau+1)}}B_j> \lceil \tilde{s}^{(\tau+1)}_j\rceil+k-1\}$\;
		
		\textit{Comment: By \cite{nguyen2015assignment} there must exist some $j \in G^{(\tau)}$ such that $
		\sum_{i\in N}\sum_{B\in \mathcal{B}_i^{(\tau+1)}}B_j\leq \lceil \tilde{s}^{(\tau+1)}_j\rceil+k-1.
		$}
		}
		
		$\tau=\tau+1$.
}
\caption{Iterative Rounding}
\label{alg:ir}
\end{algorithm}

The main difference between our proposed iterative rounding procedure and that in \cite{nguyen2015assignment} lies in the second step where we set $N^{(\tau)}$. In our procedure, we preserve the equality constraints so that we can ensure that the set of active demand 
%{\color{red}(It seems like we are keep demand constraints and throwing out supply constraints.)} {\color{blue} (Here I mean that the demand constraints that are active at $z$ is also active at $z^{(\tau)}$)} 
constraints for input $z$ is a subset of that of output $z^{(\tau)}$. Therefore we can ensure that the Lagrange multipliers of the original problem (LIP) also apply to $z^{\tau}$. This will be discussed further in Theorem \ref{thm: main} in Section \ref{sec: lc}. It then promotes the existence of \emph{supporting prices} for the allocation scheme $z^{(\tau)}$, which would be discussed in Theorem \ref{thm:supporting_price}.

It is easy to see that the procedure would finish in polynomial time, since in each iteration, either at least one variable is eliminated or at least one constraint is eliminated. So the number of rounds cannot be larger than sum of the number of constraints and the number of variables.

Define the approximate supply constraint
\begin{align*}
\sum_{i\in N}\sum_{B\in \mathcal{B}}B_jx_i(B)\leq s_j+k-1~~~~~~\forall j\in G.\tag{Supply+\emph{k}-1}
\end{align*}

We have the following results.

%\textbf{Theorem 2.1}
\begin{theorem}
	\label{thm: 2.1}
 For any reward vector $c$ and any vector $z$ that satisfies (Supply), we have $\overline{z}=\mathrm{IR}(z; c)$ to satisfy (Supply+k-1) and $c^T\overline{z}\geq c^Tz$.
\end{theorem}

\begin{proof}
See Appendix B in \cite{nguyen2015assignment}.
\end{proof}

From Theorem \ref{thm: 2.1}, we know that for any $z$ that satisfies (Demand) and (Supply) and any vector $c$, we can always find out an integral $\overline{z}$ (using Iterative Rounding Algorithm) which satisfies (Demand) and (Supply+\emph{k}-1) such that $c^T\overline{z}\geq c^Tz$.

\subsection{Lottery Construction}
\label{sec: lc}

For the Lottery Construction part, we also basically follow the algorithm described by \cite{nguyen2015assignment}, but here again we make some modifications. 

The procedure is as described in Algorithm \ref{alg: lc}.

\begin{algorithm}[!ht]
	\SetAlgoNoLine
	\KwIn{Fractional Optimal Solution $x^*$}
	\KwOut{Set $F$ of integral points that satisfy (Demand) and (Supply+$k$-1).}
	$c_{i}^{(1)}(B)=w_i(B)u_i(B),\forall B\in\mathcal{B}$\;
	
	$x'=\mathrm{IR}(x^*;c^{(1)})$\;
	
	$c_{i}^{(2)}(B)=-w_i(B)u_i(B),\forall B\in\mathcal{B}$\;
	
	$x'=\mathrm{IR}(x^*;c^{(2)})$\;
	
	$F^{(0)}=\{x',x''\}$, $\tau=0$\;
	
	$\varepsilon$ is a preset allowable error\;
	
	$\delta_z=\dfrac{\delta_\epsilon}{\sqrt{|N|\sum_{B\in\mathcal{B}}B_j^2}}$ 	
	for some $j\in G$\;
	
	\textit{Comment: Note that the value of $\delta_z$ is independent of $j$, since $\mathcal{B}$ is the set of all $k$-bundles and it is symmetric with regard to any $j\in G$.}
	
	\Repeat{Forever}{
		$y^*=\arg\min_y\{\|y-x^*\|:$ $y$ lies in the convex hull $E$ of all points in $F^{(\tau)}\}$\; (Quadratic Programming)
		
		\eIf{$\|y^*-x^*\|<\varepsilon$}{\textbf{break}}
		{
		($y^*$ may lie on a surface of the convex hull $E$. There exist a minimal subset $F'\subseteq F^{(\tau)}$ such that $y\in \mathrm{int}E'$  where $E'$ is the convex hull of $F'$.)
		
		$F^{(\tau+1)}=F'$\;
		
		$z=x^*+\delta_z\frac{x^*-y^*}{\|x^*-y^*\|}$. %Apply Iterative Rounding Algorithm on $z$ given reward vector $c=x^*-y^*$. After that we get $\overline{z}$ such that $\overline{z}^T(x^*-y^*)\geq z^T(x^*-y^*)$.
		
		$\overline{z}=\mathrm{IR}(z, x^*-y^*)$\;
		
		$F^{(\tau+1)} = F^{(\tau+1)}\cup \{\overline{z}\}$\;}
		
		$\tau=\tau+1$
	}
	\caption{Lottery Construction}
	\label{alg: lc}
\end{algorithm}

The main difference between our proposed procedure and that of \cite{nguyen2015assignment} lies in the initialization steps. In the two steps we initialize a non-empty set $F^{(0)}$ so that the quadratic programming problem in the first step in the loop would always have a solution.

The authors of \cite{nguyen2015assignment} proved that the above algorithm terminates in polynomial time. After the algorithm terminates, we get a the set of (integral) points $F$, where the fractional solution $x^*$ of (LIP) is contained in the convex hull of all points in $F=\{x^1, x^2,\cdots, x^l\}$. The coefficients $\lambda_i$ where
$
x^*=\sum_{i=1}^l \lambda_i x^i,~\sum_{i=1}^l\lambda_i=1
$
is also calculated from the quadratic programming problem in the loop. Then we can randomly select one vector $x^i$ from set $F$ (with probability $\lambda_i$) as the final allocation.

%\textbf{Theorem 2.2}
\begin{theorem}
\label{thm: main}
Suppose $x^*$ is a solution for (LIP). Any $\overline{x}\in F$ is an optimal solution to the following problem.

\begin{align}
\begin{array}{ccll}
&\displaystyle\max_{x\geq 0}&\displaystyle\sum_{i\in N}\sum_{B\in\mathcal{B}}w_i(B)u_i(B)x_i(B)&\\
&\mathrm{s.t.}&\displaystyle\sum_{B\in\mathcal{B}}x_i(B)\leq 1&\forall i\in N,\\
&&\displaystyle\sum_{i\in N}\sum_{B\in \mathcal{B}}B_jx_i(B)\leq \overline{s}_j&\forall j\in G,
\end{array}
\tag{MLIP}
\end{align}

where
\begin{align*}
\overline{s}_j:=\begin{cases}
\tilde{s}_j&\displaystyle\sum_{i\in N}\sum_{B\in \mathcal{B}}B_jx^*_i(B)< \tilde{s}_j ~\mathrm{and}~\displaystyle\sum_{i\in N}\sum_{B\in \mathcal{B}}B_j\overline{x}_i(B)\leq \tilde{s}_j\\
\displaystyle\sum_{i\in N}\sum_{B\in \mathcal{B}}B_j\overline{x}_i(B)&\mathrm{otherwise}
\end{cases}
\end{align*}
for all $j\in G$, and we have $\overline{s}_j\leq s_j+k-1$ for all $j\in G$. Furthermore, any dual solution of (LIP) is also a dual solution of (MLIP).

\end{theorem}

\begin{proof}
See Appendix \ref{app: A}.\qed
\end{proof}

From Theorem \ref{thm: main} we see that if we modify the supply vector from $\tilde{s}_j$ to $\overline{s}_j$, then we have the allocation linear programming problem to have an integral optimal solution. The optimality of the solution for a specified problem ensures that the Lagrange multiplier associated with the original problem (LIP) is also a Lagrange multiplier for (MLIP). This is important in establishing the existence of supporting prices, which will be discussed further in Theorem \ref{thm:supporting_price}. Furthermore, the fact that $\overline{s}_j\leq s_j+k-1$ indicates that for any final allocation scheme, we need no more than $k-1$ additional goods for each type of good to fulfill the allocation just as in \cite{nguyen2015assignment}.

\subsection{POPT Prices}

In this section, we will construct a set of prices from the dual solution of (MLIP) in Theorem \ref{thm: main}. Then we will discover some nice properties of the prices. We will prove that the prices \emph{support} the allocation scheme and are \emph{envy-free}. We will call these prices POPT prices.

\subsubsection{Dual Solution as Supporting Prices}
~~~~After getting the allocation from the above procedure, we would like to find a set of prices for different kind of goods to support the allocation. We would like the prices to have the property that every agent chooses the bundle that yields the best (within some acceptable error) payoff, resulting in the designed allocation.
\iffalse
and the market with $s_j$ be the amount of supply for each good $j\in G$(call it $s_j$-market) would be cleared by the agents without any race between agents, resulting in the designed allocation $\overline{x}$.\cite{easley2010networks}
\fi

\begin{definition}
Let $\epsilon_u$ be the acceptable error\footnote{For example, when people make trades of millions of dollars, they would not care the utility difference of only one dollar. In this case $\epsilon_u=1$ dollar. In general $\epsilon_u$ can be set arbitrarily small.} in utility, and $P(B)$ denote the price of bundle $B$, then a set of prices is said to \emph{support} allocation scheme $\overline{x}$ if for every agent, either the agent chooses some bundle $\hat{B}$ according to $\overline{x}$ and
\begin{align*}
u_i(\hat{B})-P(\hat{B})\geq u_i(B)-P(B)-\epsilon_u,~~~\forall B\in\mathcal{B},
\end{align*}
or he chooses to buy nothing and
\begin{align*}
u_i(B)-P(B)\leq \epsilon_u,~~~\forall B\in\mathcal{B}.
\end{align*}
\end{definition}
If $\epsilon_u=0$, then the definition coincides with users with allocations in their preferred set of bundles (after accounting for the price), and users with no allocations making a non-positive return on every good. Market-clearing prices have this property along with the market-clearing property.

\begin{theorem}
\label{thm:supporting_price}
\iffalse
If $u_i(B)>0$ for all $i\in N$ and all (non-empty) $B\in \mathcal{B}$, $u_i(B^1)<u_i(B^2)$ for all $B^1\prec B^2$ and $i\in N$, and $\sum_{j\in G}s_j'\leq k\cdot|N|$ holds, then
\fi
The set of supporting prices associated with any designed allocation $\overline{x}\in F$ (from lottery construction procedure of POPT mechanism) exists.
\end{theorem}

\begin{proof}
See Appendix \ref{app: proofsupport}.
\end{proof}

\begin{remark}
From Theorem \ref{thm:supporting_price}, we know that supporting prices always exists for the allocation we get from Lottery Construction procedure, since the solution of dual of (MLIP) are supporting prices. From the proof of the Theorem \ref{thm: main} (in Appendix \ref{app: A}), we can see that the Lagrange multipliers of original (LIP) are also multipliers of (MLIP), which means that a dual solution of (LIP) is also a dual solution of (MLIP). Therefore a dual solution of (LIP) forms a supporting price for \emph{any} designed allocation scheme $\overline{x}\in F$. Ensuring that this property holds is an important contribution of our paper.
\end{remark}

\subsubsection{Dual Solution as Approximately Envy-free Prices}
~~~~The following part shows that under the prices given by the dual solution of (MLIP), the allocation suggested by $\overline{x}$ is approximately envy-free. We first give a formal definition of envy-freeness, then we show the allocation and prices we obtain ensure this property. It is important to note that unlike in \cite{nguyen2015assignment}, envy-freeness is not added as an explicit constraint in (LIP).

\begin{definition}
Let $P(B)$ denote the price of bundle $B$ (the price of an empty bundle is 0), $x$ be the allocation vector, and $\epsilon_u$ be the acceptable error. Then the pricing rule is said to be \emph{approximately envy-free} if\footnote{The definition is an extension of the definition in \cite{guruswami2005profit} as they consider exact envy-freeness, i.e. the case when $\epsilon_u=0$.}
\begin{align*}
\sum_{B\in\mathcal{B}}[u_i(B)-P(B)]x_i(B)\geq \sum_{B\in\mathcal{B}}[u_i(B)-P(B)]x_k(B)+\epsilon_u,~~~\forall i,k\in N
\end{align*}
\end{definition}

\begin{theorem}
\label{thm: envyfree}
Under any optimal allocation $x^*$ for problem (LIP), if the pricing rule is
\begin{align*}
P(B)=\sum_{j\in G}B_jp_j^*
\end{align*}
where $p_j^*~(j\in G)$ is some Lagrange multiplier associated with supply constraints of (LIP), then the approximate envy-freeness condition holds.
\end{theorem}

\begin{proof}
See Appendix \ref{app: envyfreeproof}.
\end{proof}
\iffalse
\textbf{Proposition 2.5} Under the allocation $\overline{x}$ resulting from Iterative Rounding and any market clearing price associated with $\overline{x}$, the envy-freeness condition holds for the given price rule.

\emph{Proof.}  Since under market clearing prices, agents would make choices according to individual rationality. It is a direct result from individual rationality that every agent would choose among the (possibly empty) bundles that yields the maximum payoff. 

Since $\overline{x}$ is an optimal solution of (MLIP), 
\qed
\fi

%\textbf{Proposition 2.5}
\begin{proposition}
\label{prop: 2.5}
For any integral allocation $\overline{x}\in B$, if the price rule is
\begin{align*}
P(B)=\sum_{j\in G}B_jp_j^*
\end{align*}
where $p_j^*~(j\in G)$ is some Lagrange multiplier associated with supply constraints of (LIP), then the approximate envy-freeness condition (not only holds for $x^*$ but) also holds for $\overline{x}$.
\end{proposition}

\begin{proof}
From the Proof of Theorem \ref{thm: main} (Appendix \ref{app: A}) we know that if $(\pi^*, v^*)$ is a set of Lagrange multiplier of (LIP), then $(\pi^*, v^*)$ is also a set of Lagrange multiplier of (MLIP). Since $\overline{x}$ is a solution to (MLIP) according to Theorem \ref{thm: main}, we have $p^*_j$ to be also the Largrange Multipliers for (MLIP). Apply Theorem \ref{thm:supporting_price} to (MLIP) then we would find that the approximate envy-freeness condition holds for $\overline{x}$.
\end{proof}

From Proposition \ref{prop: 2.5} we know that we can calculate the dual solution of (LIP) instead of (MLIP) to figure out a set of POPT prices that supports all possible allocation $\overline{x}$. This is important, as we can then get the dual solution of (LIP) through the process of solving (LIP). We do not need to construct and solve the dual of (MLIP), which makes the calculation of prices much more efficient.

\subsection{Summary of POPT Mechanism}

For the POPT mechanism, we try to find an approximately efficient and envy-free allocation scheme for spectrum band allocation with complementarities as well as the price of goods associated with the scheme. To achieve this we first relax the constraint to allow the allocation indicator variables to be real-valued, and we set up an initial linear programming problem (LIP) with only the demand and supply constraints. We then solve (LIP) to get a solution $x^*$. 

Then, we perform the Lottery construction process on $x^*$ to round $x^*$ to a set $F$ of integral solutions. We then assign a probability to each of these integral solutions to form a lottery such that it has expectation $x^*$.

When realizing the mechanism, we choose an integral solution in $F$ based on some random event (such as rolling a die) so that every integral solution has the probability of being chosen equal to the probability assigned to it.

We solve the dual problem of (LIP) at the same time to obtain the prices associated with POPT mechanism. As implied by Theorem \ref{thm:supporting_price} and Theorem \ref{thm: envyfree}, this price supports any allocation scheme given by POPT mechanism and it is approximately envy-free.

%% file: analysis.tex
\section{Analysis}\label{sec:analysis}
In this section we start by showing asymptotic strategy-proofness of the POPT mechanism, and then present a numerical analysis of a specific spectrum allocation problem.

\subsection{Theoretical Analysis}

In this part, we will first analyze the complexity and strategy-proofness of the mechanism. Then we will discuss on the methods for solving the linear programming problems for the mechanism.
\subsubsection{Complexity Analysis}
~~~~Fix $k$, the calculation of the approximate optimal allocation can be performed in polynomial time with regard to $|N|$ and $|G|$. (See Appendix \ref{app: B}) The mechanism is therefore much more computationally efficient than the optimal mechanism: integer linear programming with VCG prices.
\subsubsection{Asymptotic Strategy-proofness}
~~~~One important property that a mechanism should have is strategy-proofness. A mechanism is said to be \emph{strategy-proof} if the best strategy for any individual agent is to report his or her utility truthfully.

Consider that we have finite types of agents, and the set of type is denoted as $\Theta=\{\theta_1, \theta_2,\cdots,\theta_m\}$, and the number of agents with type $\theta$ is $n_\theta$. Denote the utility of bundle $B$ for agent of type $\theta$ as $u_\theta(B)$ and the type of agent $i$ as $\beta_i$. Following \cite{nguyen2015assignment} we define {asymptotic strategy-proofness} as follows.

\begin{definition}
If $\overline{x},\overline{p}$ are the allocation vector and price vector that POPT mechanism gives when every agent report his or her type $\beta_i$ truthfully, and $\tilde{x}$ and $\tilde{p}$ are the corresponding results when agent $i$ report as type $\gamma$ while others report their types truthfully, a mechanism is called \emph{asymptotically strategy-proof} if for any $\varepsilon_0>0$ there exists $N_0\in\mathbb{N}$ such that when $n_\theta>N_0$ for all $t=1,2,\cdots,m$ we have
\begin{align*}
\mathbb{E}\left[\sum_{B\in \mathcal{B}}\left(u_{\beta_i}(B)-\overline{p}(B)\right)\overline{x}_i(B)\right]\geq \mathbb{E}\left[\sum_{B\in \mathcal{B}}(u_{\zeta}(B)-\tilde{p}(B))\tilde{x}_i(B)\right]-\epsilon_u-\varepsilon_0
\end{align*}
for all $i\in N,\zeta\in \Theta$, where $\epsilon_u>0$ is the acceptable error.
\end{definition}

This definition generally means that when the number of agents comes large, the increase of average payoff that an agent could get by misreporting his or her type becomes less and less and converges to a number less than or equal to $\epsilon_u$. If the agents are risk-neutral, then this implies that an agent would not have an impulse to mis-report her type.

The following theorem is a variant of Theorem 3.3 in \cite{nguyen2015assignment}. The additional steps arise from ensuring the convergence of the dual variables.

\begin{theorem}
\label{thm: sp}
Set $w_i$ to be the same for the same type of agents, then the mechanism POPT is asymptotically strategy-proof.
\end{theorem}

\begin{proof}
See Appendix \ref{app: C}.
\end{proof}

Therefore, we can say that the POPT mechanism is approximately truthful. When the number of agents get large, one can get approximately the best payoff by reporting his or her true type.

\subsubsection{Methods for Solving Linear Programming Problem}
~~~~For the initial linear programming problem (LIP) and updated linear programming problems (ULIP) in Iterative Rounding, we choose the Simplex method to solve them. The reasons for us to choose Simplex method are given as follows.
\begin{enumerate}
	\item It is fast in practice. The average complexity of simplex method is $O(s)$ where $s$ is the number of constraints. In (LIP), we have significantly fewer constraints than variables, thus it is very suitable for (LIP).
	\item It solves the problem based on extreme points. In (LIP) and all (ULIP), we would like to find an extreme point solution. Simplex method ensures that the solution is an extreme point.
\end{enumerate}

\subsection{Empirical Analysis}

In the following part, we will analyze the whole mechanism from multiple perspectives. Based on a utility model that is described in \cite{zhou2013complexity}, we will first investigate the efficiency of the mechanism on a utility model. Then we will investigate how many more goods we need on average to fulfill the allocation suggested by the mechanism. We will use a grid structure for the analysis, so next we will discuss how the prices and the allocations change based on location. Finally, we finely categorize the form of integer solutions produced by our mechanism, especially the shapes of the bundles.

The utility model is constructed in the following way: Consider a 2-D area that is a $m_g\times n_g$ grid, each grid has unit area and it is considered as a type of good. For each grid (good) there are $s_g$ bands, which can be viewed as the available supply for the grid. There are $N_a$ agents and each agent would like at most a bundle of $k_a$ goods.

For each agent, there are some end-users distributed via a 2-D spatial Poisson process in the area, with parameter $\mu$ (person/unit area). The parameter $\mu$ determines the system load, which we will vary. The utility of grid $j$ to an agent $i$, which is denoted as $u^i_j$, is proportional to the number of end users associated with that agent in that area. However, this utility can only be fully realized when the agent gets the same amount of bands in all adjacent areas. If the agent fails to get the same amount of bands in some adjacent area, a boundary cost applies.

The boundary cost if agent $i$ gets some bands in area $j$ but fails to get the same amount of bands in area $k$ is denoted as $c_{jk}^i$. We have $c_{jk}^i$ to be proportional to the number of end users of agent $i$ in the boundary area of grid $j$ which is close to grid $k$. For each grid, we specify $\lambda$ as the proportion of boundary area. Denote $G$ as the set of grids (goods), and $E$ as the set of adjacent grid pairs. Then the utility functions on bundles are given by
\begin{align*}
u_i(B)=\sum_{j\in G}u^i_j-\sum_{(j,k)\in E}(B_j-B_k)^+c_{jk}^i
\end{align*}
We do not specialize our analysis and mechanism to the given form of the utility functions. In future work, we plan on studying mechanisms tailored to the specific forms of utility functions described above in order to see if further improvements in performance (in terms of complexity and memory) can be had.

\subsubsection{Analysis of Efficiency}
~~~~Set $m_g=n_g=3, s_g=10, N_a=30, k_a=4$. The total utility versus the portion $\lambda$ of boundary is shown in Figure \ref{fig: 1}. From the figure we can see when the boundary area gets larger, the boundary cost rises and the total utility gets smaller. 

The total utility under three different allocations are compared: average utility for the integral LP optimal solution when the supplies are $s_j+k-1$; average utility under POPT mechanism (when supplies are $s_j$); and average utility for the (fractional) LP optimal solution (when supplies are $s_j$). The experimental results match the theory so that the POPT mechanism has nearly the same average utility as LP optimal solutions. The average for IntLP is only a little bit larger than that of the POPT mechanism for all values of $\mu$ and $\lambda$, which indicates that POPT mechanism is near optimal in expectation. Note that we have to consider the average utility in order to average the spatial Poisson process realizations.

\begin{figure}[!ht]
	\centering
	\includegraphics[width=8cm]{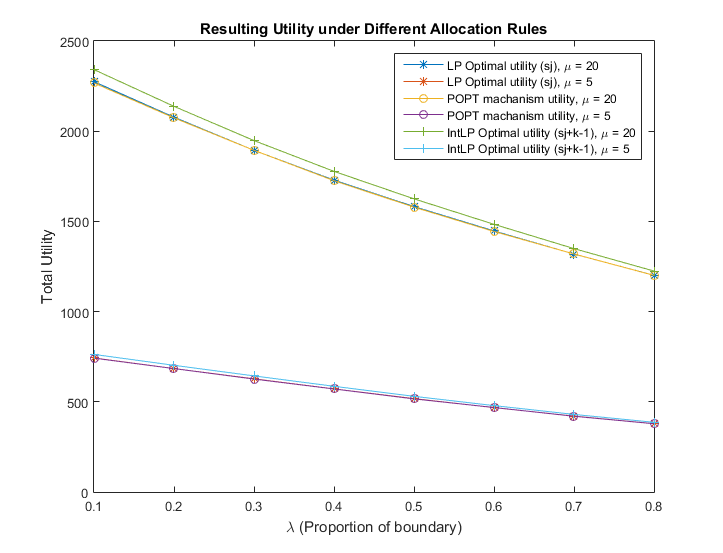}
	\caption{Resulting utility under different allocation rules vs. $\lambda$ for $m_g=n_g=3, s_g=10, N_a=30, k_a=4$.}
	\label{fig: 1}
\end{figure}

%\subsubsection{Analysis of Complimentarity}
\subsubsection{Analysis of Overallocation of Supply}
~~~~From analysis of the POPT mechanism it is clear that we may need to add some additional goods for some regions. We have seen that the theoretical bound for the number of additional goods added per supply is $k-1$. However,in the case $m_g=n_g=3, s_g=10, N_a=30, k_a=4$ (where the total demand exceeds the total supply), experimental results (Figure \ref{fig: addgood}) show that in most cases, the total number of additional goods is significantly less than $(k-1)|G|$. While we have $(k_a-1)|G|=27$, we have the total number of additional goods less than or equal to 12 in more than $99\%$ of the cases. 

\begin{figure}[!ht]
	\centering
	\includegraphics[width=5cm]{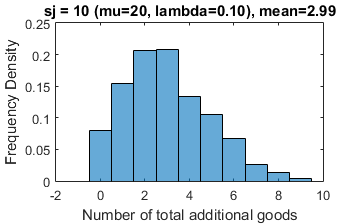}
	\includegraphics[width=5cm]{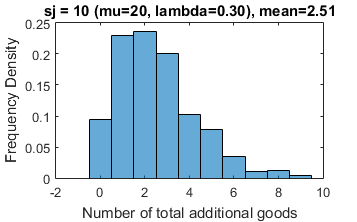}\\
	
	\includegraphics[width=5cm]{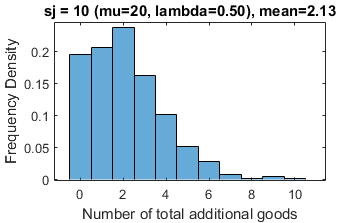}
	\includegraphics[width=5cm]{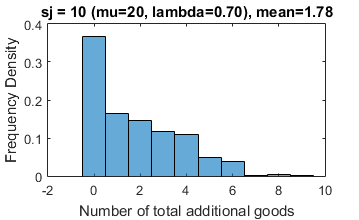}
	\caption{Number of over-allocated goods when $m_g=n_g=3, s_g=10, N_a=30, k_a=4, \mu=20$.}
	\label{fig: addgood}
\end{figure}

From Figure \ref{fig: addgoodvsla}, we can see that the average additional goods gets smaller as the interference cost coefficient $\lambda$ gets larger. When the interference cost is low ($\lambda=0.1$), we have the average additional goods to be 2.99, which then decreases to roughly 1.5 when the interference cost is high ($\lambda=0.8$). Both of these are a small compared to the total number of goods, $m_gn_gs_g=90$. Even the maximum number of goods added, empiricaly $10$, is a small fraction of the total number of good. We then conclude that the mechanism is approximately efficient, and it does not require the addition of significantly more goods for each type.

\begin{figure}[!ht]
	\centering
	\includegraphics[width=6.5cm]{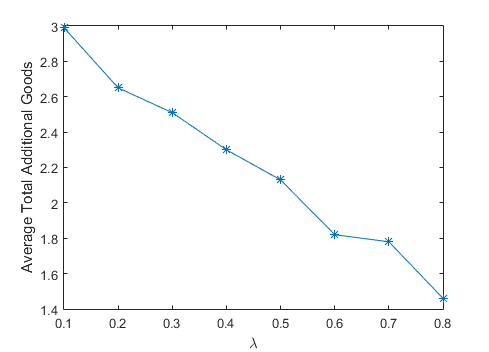}
	\caption{Average Additional Goods vs. $\lambda$ for $m_g=n_g=3, s_g=10, N_a=30, k_a=4, \mu=20$}
	\label{fig: addgoodvsla}
\end{figure}

\subsubsection{Analysis of Local Statistics}
~~~~Besides the global performance of the mechanism, such as the total utility, total additional goods, we also look into the statistics of the results for each element of the grid. For the case $m_g=n_g=3, s_g=10, N_a=30, k_a=4, \mu=20$, we investigate the distribution of number of actual allocated goods for each element of the grid. Here the numbers are such that the total demand exceeds the total supply.

\iffalse
We can see in Figure \ref{fig: X} that, the distribution of the number of allocations does not differ a lot with regard to the location of the grid. As $\lambda$ gets larger, the distribution is more centered at $s=10$. 
\fi
Since the allocation distributions for each element of the grid are visually similar looking, we compute the total variational distances between the distribution of different grid locations. These distances are used as a metric for us to determine whether two distributions are similar or not. For the 3x3 grid, grid locations $(1,1)$, $(1,3)$, $(3,1)$ and $(3,3)$ are statistically similar, as they are both on the corners of the area and share the same number of borders with the neighboring locations. By the same logic, grid locations $(1,2)$, $(2,1)$, $(2,3)$ and $(3,2)$ are also statistically similar. For statistically similar grids, we expect that the distributions of the allocations in these grid locations would be close to each other; note that this argument is only for the marginal distributions. In Figure~\ref{fig: X} we find that the distance between allocation distributions on statistically dissimilar grid locations does not differ a lot from that of statistically similar ones. We believe this because the demand exceeds the supply so that all locations get fully subscribed.
%Therefore, we compute the total variation distance of statistically symmetrical grids as baselines to be compared with that between statistically asymmetrical grids. However when $\lambda$ is fixed the grids at statistically symmetrical locations does not show a significant difference in distribution. The distance between distributions on statistically asymmetrical grids does not differ a lot from that of statistically symmetrical grids.

\begin{figure}[!ht]
	\centering
	\includegraphics[width=7cm]{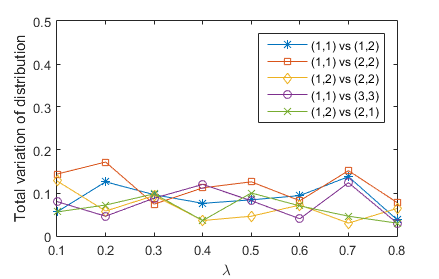}
	\caption{Total variation distance of allocation distributions vs. $\lambda$ for $m_g=n_g=3, s_g=10, N_a=30, k_a=4, \mu=20$}
	\label{fig: X}
\end{figure}

However, as we notice from Figure \ref{fig: P} when the total demand exceeds the total supply, the prices at different grid locations show significant variations. The average prices of interior grid locations are generally higher than that of the grid locations on the boundary of the arena. We also observe that the larger the interference parameter $\lambda$, the larger the difference in the average prices. This matches our intuition that the interior grids are generally more valuable since there are fewer boundaries with other locations, especially when the interference cost is large.

\begin{figure}[!ht]
	\centering
	\includegraphics[width=6.5cm]{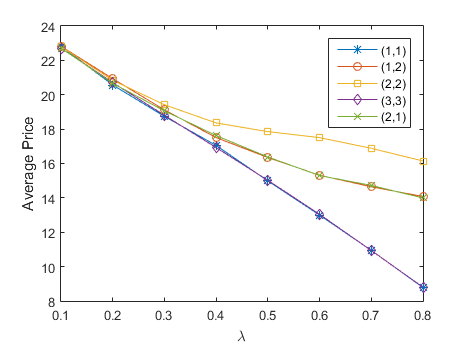}
	\caption{Average price of different grids vs. $\lambda$ for $m_g=n_g=3, s_g=10, N_a=30, k_a=4, \mu=20$}
	\label{fig: P}
\end{figure}

For the case $m_g=n_g=4, s_g=10, N_a=30, k_a=4, \mu=20$, we also investigated the distribution of number of allocation on each grid. Here the total supply exceeds the total demand. Unlike the case when total demand exceeds the total supply, here we can see that the number of allocations differs significantly based on the grid locations (Figure \ref{fig: Y}). In general, the average number of allocations for interior grids are greater than that of the grids on the boundaries of the area. Here we do not display any statistics for the prices as they're all uniformly low, as is to be expected.
\iffalse
\begin{figure}[!ht]
	\centering
	\includegraphics[width=4cm]{lc41.png}
	\includegraphics[width=4cm]{lc42.png}
	\includegraphics[width=4cm]{lc43.png}\\
	\includegraphics[width=4cm]{lc51.png}
	\includegraphics[width=4cm]{lc52.png}
	\includegraphics[width=4cm]{lc53.png}\\
	\includegraphics[width=4cm]{lc61.png}
	\includegraphics[width=4cm]{lc62.png}
	\includegraphics[width=4cm]{lc63.png}\\
	\caption{Number of allocations for specific grid locations for $m_g=n_g=4, s_g=10, N_a=30, k_a=4$}
	\label{fig: Y}
\end{figure}
\fi

\begin{figure}[!ht]
	\centering
	\includegraphics[width=6cm]{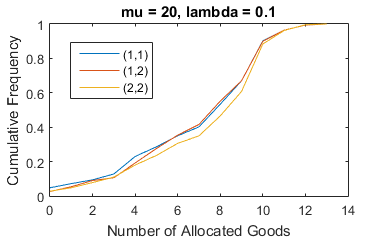}\\
	\includegraphics[width=6cm]{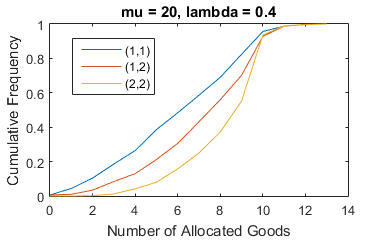}\\
	\includegraphics[width=6cm]{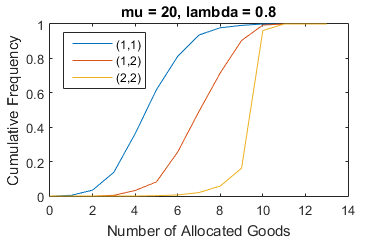}
	\caption{CDF of Number of allocations for specific grid locations for $m_g=n_g=4, s_g=10, N_a=30, k_a=4$}
	\label{fig: Y}
\end{figure}

\subsubsection{Analysis of Shape of Bundles}
~~~~We also investigate on the geometry of the bundles that the mechanism allocate to the users. Typically, bundles with more internal boundaries are more valuable than those with fewer internal boundaries. We expect that the mechanism would prefer to allocate bundles with more internal boundaries (such as "O"-shaped 4-bundles), especially when the interference cost is high.

The empirical results agree with our intuition. As we observe from Figure \ref{fig: S}, for the case where $m_g=n_g=3, s_g=10, N_a=30, k_a=4, \mu=20, \lambda=0.8$, more 4-bundles are allocated than 3-bundles. Note that this is a setting where the demand exceeds the supply. Among all the 4-bundles that are allocated, bundles with more internal boundaries are significantly preferred by the mechanism, which means that the number of "O"-shaped bundles of the resulting allocation is significantly higher than that of bundles of any other shape ("L"-shaped, "T"-shaped, or "Z"-shaped bundles). 

\begin{figure}[!ht]
	\centering
	\includegraphics[width=5cm]{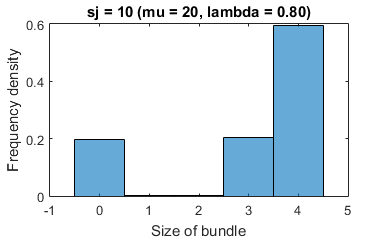}
	\includegraphics[width=5cm]{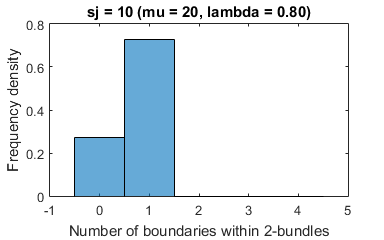}\\
	\includegraphics[width=5cm]{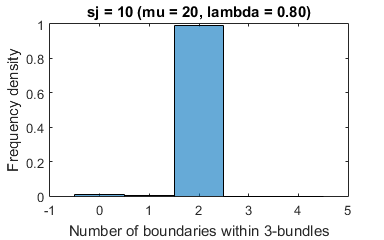}
	\includegraphics[width=5cm]{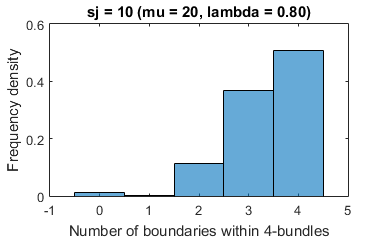}
	\caption{Statistics on size of bundles, and number of boundaries within bundles given the size of a bundle, $m_g=n_g=3, s_g=10, N_a=30, k_a=4,\lambda=0.8$}
	\label{fig: S}
\end{figure}

We also investigate the case when $m_g=n_g=4, s_g=10, N_a=45, k_a=4, \mu=20, \lambda=0.8$. In this case, nearly all assigned bundles are 4-bundles. Note that here the total supply exceeds the total demand. From Figure \ref{fig: S1} we can see that most of the assigned 4-bundles has 4 internal boundaries, which means that they are "O"-shaped. This again matches our intuition.

\begin{figure}[!ht]
	\centering
	\includegraphics[width=5cm]{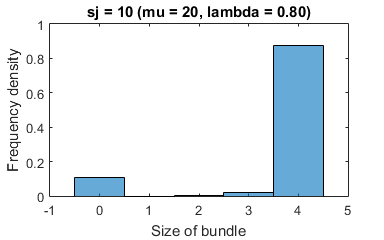}
	\includegraphics[width=5cm]{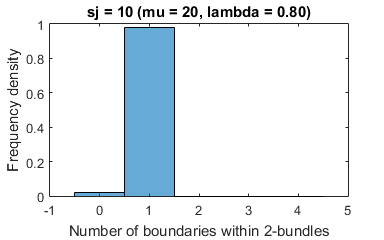}\\
	\includegraphics[width=5cm]{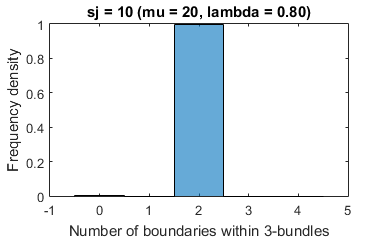}
	\includegraphics[width=5cm]{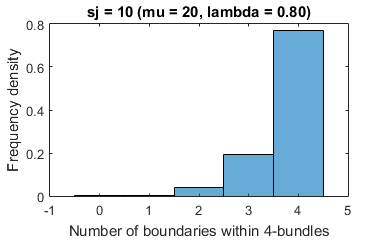}
	\caption{Statistics on size of bundles, and number of boundaries within bundles given the size of a bundle, $m_g=n_g=3, s_g=10, N_a=45, k_a=4,\lambda=0.8$}
	\label{fig: S1}
\end{figure}

\subsubsection{Analysis of Supporting Prices}
~~~~We have theoretical guarantees that every agent obtains a bundle with his or her best possible payoff within some error $\epsilon_u$. The error $\epsilon_u$ can be controlled by parameter $\delta_w$. To verify this empirically we calculate the difference between the actual payoff an agent obtains and the best possible payoff. Theoretically, this difference should be larger than $-\epsilon_u$, which is a very small number. Setting $\delta_w=10^{-5}$, in the case of $m_g=n_g=3, s_g=10, N_a=45, k_a=4, \mu=20, \lambda=0.8$, we observe in Figure \ref{fig: D} that in most cases the difference is 0; results of all iterations and all agents are pooled together in the figure. We also find that the minimum payoff difference we obtain in simulation is $-0.0011$, which indicates a loss. However as we found the average actual payoff that agents obtain to be $4.123$, we conclude that the loss is negligible. Actually, this loss can be arbitrarily small as we set $\delta_w$ to be small.

\begin{figure}[!ht]
	\centering
	\includegraphics[width=10cm]{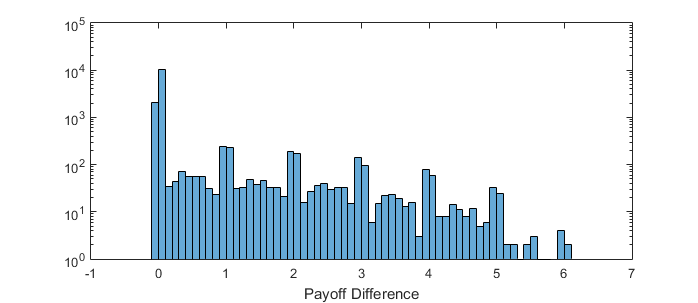}
	\caption{Payoff difference of agents for $m_g=n_g=3, s_g=10, N_a=45, k_a=4, \mu=20, \lambda=0.8$}
	\label{fig: D}
\end{figure}

%% file: implementation_conclusion.tex
\section{Implementation}\label{sec:implement}

We implemented the POPT mechanism as the Spectrum Allocation Tool in Python. 
%to realize the POPT mechanism. 
It takes advantage of package CVXPY \cite{cvxpy} and uses ECOS solver \cite{ecos} to solve linear programming and quadratic programming problems, where ECOS is one of the most efficient solvers among all free solvers. The tool can be easily adapted to commercial solvers like GUROBI or MOSEK when these solvers are available to the user.
%
%\begin{figure}[!ht]
%	\centering
%	\includegraphics[width=10cm]{gui.png}
%	\caption{Spectrum Allocation Tool GUI}
%\end{figure}
%
The tool accepts two kinds of input: JSON and plain text. GUI input is currently not accepted as the input is usually some large data, which makes GUI input inconvenient. After calculation, the tool will visualize the resulting allocation scheme.
The tool is open-source and is available at \url{https://bitbucket.org/dwtang/batool} .

\section{Conclusion and Future Improvements}\label{sec:conclusion}
\iffalse
~~~~Nguyen et al.\cite{nguyen2015assignment} has casted an insight on envy-free mechanisms for agents with complementarities in preferences.
\fi
In this paper we proposed the POPT mechanism based on the OPT mechanism for matching problems where prices cannot be assessed \cite{nguyen2015assignment}, where we utilize prices to allocate goods to agents with complementarities in preferences. We have proved that the mechanism is approximately envy-free and asymptotically strategy-proof. We have also shown that the mechanism is approximately efficient in practice. Finally we implemented the mechanism as a GUI tool to enable practice of the mechanism in real world settings. 

The POPT mechanism has a few nice properties. %Firstly, the mechanism utilizes prices to ensure that agents report their values truthfully, and the price is approximately envy-free.
First, as the mechanism utilizes prices, envy-freeness emerges as a consequence and does not need to be imposed as a constraint in the underlying optimization problem.
Furthermore, the prices for the mechanism are independent of what allocation the lottery scheme picks from the available allocation scheme set. Since the prices are the dual variables of the initial linear programming relaxation, we can calculate the prices at the same time of solving it, which makes the calculation of prices efficient.

In the future work, we plan on generalizing our mechanism to the case that each agent may desire more than one bundle. We also plan on determining theoretic guarantees on the efficiency of the mechanism using the optimal integer allocation with $k-1$ extra goods of each type.

%% file: paper.bbl

%% file: appendixA.tex
\section{Proof of Theorem 2.2}
\label{app: A}
Let $x^*$ be one optimal solution for (LIP). Define set $S$ as
\begin{align*}
S=\Big\{&z\in \mathbb{R}^{|N|\times|\mathcal{B}|}~\Big|~x^*_i(B)=0\Rightarrow z_i(B)=0,\\
&\sum_{i\in N}x^*_i(B)=1\Rightarrow \sum_{i\in N}z_i(B)=1,~
z~\mathrm{satisfies~(Supply}+k-1)\Big\}
\end{align*}

\begin{lemma}
\label{lemma:A.1}
For any $z\in S$ that satisfies (Supply), we have $\overline{x}=\mathrm{IR}(z)\in S$.
\end{lemma}

\begin{proof}
If $z\in S$ and $\overline{x}=\mathrm{IR}(z)$, it follows from definition of $S$ as well as the procedure of Iterative Rounding (described in Section 2.2.1) that
\begin{align*}
x^*_i(B)&\Rightarrow z_i(B)=0\Rightarrow\overline{x}_i(B)=0\\
\sum_{i\in N}x^*_i(B)=1&\Rightarrow\sum_{i\in N}z_i(B)=1\Rightarrow \sum_{i\in N}\overline{x}_i(B)=1
\end{align*}

If $z$ satisfies (Supply), from Theorem \ref{thm: 2.1} it follows that $\overline{x}$ satisfies (Supply+$k$-1) (thus we have $\overline{s}_j\leq s_j+k-1$). Therefore, for any $z\in S$ that satisfies (Supply), we have $\overline{x}=\mathrm{IR}(z)\in S$.
\end{proof}

Then we will use induction to prove that $F\subseteq S$ after Lottery Construction procedure.
\begin{itemize}
	
	\item From Lemma \ref{lemma:A.1} we deduct that since $x^*\in S$ and $x^*$ satisfies (Supply), we have $x',x''\in S$. Hence, the initial $F^{(0)}\subseteq S$.
	
	\item Suppose that in the $p$-th iteration of Lottery Construction we have $F^{(p)}\subseteq S$. In this iteration, we have $y^*=\sum_{t=1}^l \lambda_t x^t$,  $\sum_{t=1}^l\lambda_t=1$ where $x^t\in S$ for $t=1,2,\cdots,l$. As
	\begin{align*}
	x^*_i(B)=0&\Rightarrow x^1_i(B)=0, x^2_i(B)=0,\cdots, x^l_i(B)=0\\
	&\Rightarrow y^*_i(B)=\sum_{j=1}^l \lambda_j x^j_i(B)=0\\
	\sum_{i\in N}x^*_i(B)=1&\Rightarrow \sum_{i\in N}x^1_i(B)=1, \sum_{i\in N}x^2_i(B)=1,\cdots,\sum_{i\in N}x^l_i(B)=1\\
	&\Rightarrow \sum_{i\in N}y^*_i(B)=\sum_{t=1}^l \lambda_t \sum_{i\in N}x^j_i(B)=\sum_{t=1}^l \lambda_t=1
	\end{align*}
	
	and
	\begin{align*}
	\sum_{i\in N}\sum_{B\in\mathcal{B}}B_jy^*_i(B)&=\sum_{t=1}^l\lambda_t\left(\sum_{i\in N}\sum_{B\in\mathcal{B}}B_jx_i^j(B)\right)\\
	&\leq \sum_{t=1}^l \lambda_t(s_j+k-1) = (s_j+k-1)\sum_{t=1}^l \lambda_t\\
	&=s_j+k-1
	\end{align*}
	which means that $y^*$ satisfies (Supply+$k$-1). We then have $y^*\in S$.
	
	Then, as $z=x^*+\delta_z\frac{x^*-y^*}{\|x^*-y^*\|}$, we have
	\begin{align*}
	\sum_{i\in N}\sum_{B\in\mathcal{B}}B_jz_i(B)&=\sum_{i\in N}\sum_{B\in\mathcal{B}}B_jx^*_i(B)+\delta_z\sum_{i\in N}\sum_{B\in\mathcal{B}}B_j\frac{x^*_i(B)-y^*_i(B)}{\|x^*-y^*\|}\\
	&\leq \tilde{s}_j+\dfrac{\delta_\epsilon\sum_{i\in N}\sum_{B\in\mathcal{B}}B_j(x^*_i(B)-y^*_i(B))}{\sqrt{\sum_{i\in N}\sum_{B\in\mathcal{B}}B_j^2}\cdot \|x^*-y^*\|}\\
	&=s_j-\epsilon_j+\delta_\epsilon\dfrac{\overline{b}^T (x^*-y^*)}{\|\overline{b}\|\cdot \|x^*-y^*\|}\tag{$\overline{b}^T$ is the $j$-th row vector of $B$}\\
	&\leq s_j-\delta_\epsilon+\delta_\epsilon\tag{Using Cauchy-Schwarz Inequality}\\&=s_j
	\end{align*}
	Thus $z$ satisfies (Supply). And as $y^*\in S$ we have
	\begin{align*}
	x^*_i(B)=0&\Rightarrow x_i^*(B), y^*_i(B)=0\\&\Rightarrow z_i(B)=x^*_i(B)+\delta_z\frac{x^*_i(B)-y^*_i(B)}{\|x^*-y^*\|}=0\\
	\sum_{i\in N}x^*_i(B)=1&\Rightarrow \sum_{i\in N}x^*_i(B)=1, \sum_{i\in N}y^*_i(B)=1\\&\Rightarrow \sum_{i\in N}z_i(B)=\sum_{i\in N}x^*_i(B)+\delta_z\frac{\sum_{i\in N}x^*_i(B)-\sum_{i\in N}y^*_i(B)}{\|x^*-y^*\|}\\
	&~~~~~~~~~~~=1+\delta_z\cdot \frac{1-1}{\|x^*-y^*\|}=1
	\end{align*}
	Thus we have $z\in S$. Therefore by Lemma \ref{lemma:A.1} $\overline{z}=\mathrm{IR}(z;c)\in S$. Therefore, the updated ($p+1$-th iteration's) $F^{(p+1)}=F^{(p)}\cup \{\overline{z}\}$ still satisfies $F^{(p+1)}\subseteq S$.
\end{itemize}

Therefore we conclude that $F\subseteq S$ in all iterations. Thus the final set $F\subseteq S$. This means that, any allocation vector $\overline{x}\in F$ satisfies
\begin{align*}
x^*_i(B)=0&\Rightarrow \overline{x}_i(B)=0,\\
\sum_{i\in N}x^*_i(B)=1&\Rightarrow \sum_{i\in N}\overline{x}_i(B)=1
\end{align*}

From the definition of $\overline{s}_j$ we can see that
\begin{align*}
\sum_{i\in N}\sum_{B\in\mathcal{B}}B_jz_i(B)=\tilde{s}_j\Rightarrow \sum_{i\in N}\sum_{B\in\mathcal{B}}B_j\overline{x}_i(B)=\overline{s}_j
\end{align*}

Denote the coefficient matrix of $x$ in constraints of (LIP) and (MLIP) as $A$ and the upperbound vector as $b$ for (LIP) and $\overline{b}$ for (MLIP). Then (LIP) can be denoted as $\displaystyle\max_{x\geq 0}\{u^Tx: Ax\leq b\}$ and (MLIP) can be denoted as $\displaystyle\max_{x\geq 0}\{u^Tx: Ax\leq \overline{b}\}$. Then the above can be rewritten as
\begin{align*}
(Ax^*-b)_k=0\Rightarrow (A\overline{x}-\overline{b})_k=0
\end{align*}

Since linear programming problems are convex problems, Karush-Kuhn-Tucker conditions are necessary and sufficient for optimality. Since $x^*$ is optimal for (LIP) we have the KKT condition: exist $\pi^*,v^*\geq 0$ such that
\begin{align*}
-u+A^T\pi^*-v^*&=0\\
{\pi^*}^T(Ax^*-b)&=0\\
Ax^*&\leq b\\
{v^*}^Tx^*&=0
\end{align*}

By complementary slackness we have
\begin{align*}
v^*_i(B)&>0\Rightarrow z_i(B)=0\Rightarrow \overline{x}_i(B)=0\\
\pi^*_k&>0\Rightarrow (Ax^*-b)_k=0\Rightarrow (A\overline{x}-\overline{b})_k=0
\end{align*}

Therefore we have
\begin{align*}
{v^*}^T\overline{x}&=0\\
{\pi^*}^T(A\overline{x}-\overline{b})&=0
\end{align*}

And given the definition of $\overline{s}_j$ we can find that $\overline{x}$ is feasible for (MLIP), that is
\begin{align*}
A\overline{x}\leq \overline{b}
\end{align*}

Thus we have KKT conditions also holds for $\overline{x}$ in (MLIP) (and the same Lagrange multipliers apply). Therefore $\overline{x}$ is an optimal solution to (MLIP). 

As a result of the proof, we show that any dual variable of (LIP) is also a dual solution to (MLIP).\qed

%% file: appendix_proof_envyfreeness.tex
\section{Proof of Theorem 2.4}
\label{app: proofsupport}

The dual problem of (MLIP) in Theorem \ref{thm: main} is given by
\begin{align}
\begin{array}{ccll}
&\displaystyle\min_{\alpha,p\geq 0}&\displaystyle\sum_{i\in N}\alpha_i+\displaystyle\sum_{j\in G}\overline{s}_j p_j\\
&\mathrm{s.t.}&\alpha_i\geq w_i(B) u_i(B)-\displaystyle\sum_{j\in G}B_j p_j&\forall i\in N,B\in\mathcal{B}
\end{array}
\tag{DMLIP}
\end{align}
where $\alpha_i,\forall i\in N$ are the Lagrange multipliers associated with demand constraints of (MLIP), and $p_j,\forall j\in G$ are the Lagrange multipliers associated with supply constraints of (MLIP).

As $\overline{x}$ is the optimal solution to the primal problem according to Theorem \ref{thm: main}, we know that the dual problem must also have at least one optimal solution. 

Suppose that $(\overline{\alpha}, \overline{p})$ is an optimal solution for the dual problem, first we must have
\begin{align*}
\overline{\alpha}_i=\max_{B\in\mathcal{B}}\left\{w_i(B)u_i(B)-\sum_{j\in G}B_j\overline{p}_j,0\right\}.
\end{align*}
(If $\overline{\alpha}_i$ is less than the right hand side, then it is not feasible for the dual problem. If $\overline{\alpha}_i$ is larger than the right hand side, then the objective function value is larger)

By complementary slackness condition we have
\begin{align*}
\overline{x}_i(\hat{B})>0~~~&\Rightarrow~~~\overline{\alpha}_i=w_i(\hat{B})u_i(\hat{B})-\sum_{j\in G}\hat{B}_j \overline{p}_j\\
\displaystyle\sum_{i\in N}\overline{x}_i(\hat{B})<1~~~&\Rightarrow~~~\overline{\alpha}_i=0\\
\end{align*}

Then we have
\begin{align*}
&\mathrm{Agent~}i\mathrm{~gets~bundle~}\hat{B}\\\Rightarrow~~~&\overline{\alpha}_i=w_i(\hat{B})u_i(\hat{B})-\sum_{j\in G}\hat{B}_j \overline{p}_j=\max_{B\in\mathcal{B}}\left\{w_i(B)u_i(B)-\sum_{j\in G}B_j\overline{p}_j\right\}\\
\Rightarrow~~~&u_i(\hat{B})-\sum_{j\in G}\hat{B}_j\overline{p}_j\geq u_i(B)-\sum_{j\in G}B_j\overline{p}_j\\&~~~~~~~~~~~~~~~~~~~~~~-[(w_i(\hat{B})-1)u_i(\hat{B})+(1-w_i(B))u_i(B)]~~~\forall B\in \mathcal{B}
\end{align*}
\begin{align*}
&\mathrm{Agent~}i\mathrm{~gets~no~bundle~}\\\Rightarrow~~~&\overline{\alpha}_i=0=\max_{B\in\mathcal{B}}\left\{w_i(B)u_i(B)-\sum_{j\in G}B_j\overline{p}_j,0\right\}\\
\Rightarrow~~~&w_i(B)u_i(B)-\sum_{j\in G}B_j\overline{p}_j\leq 0~~~\forall B\in \mathcal{B}\\
\Rightarrow~~~&u_i(B)-\sum_{j\in G}B_j\overline{p}_j\leq (1-w_i(B))u_i(B)~~~\forall B\in \mathcal{B}
\end{align*}

Define $M_i=\max_{B\in\mathcal{B}}\{u_i(B)\}$. Since $w_i$ is uniformly distributed in $[1-\delta_w, 1+\delta_w]$ we have the following
\begin{align*}
&\mathrm{Agent~}i\mathrm{~gets~bundle~}\hat{B}\\
\Rightarrow~~~&u_i(\hat{B})-\sum_{j\in G}\hat{B}_j\overline{p}_j\geq u_i(B)-\sum_{j\in G}B_j\overline{p}_j-2\delta_w M_i~~~\forall B\in \mathcal{B}\\
&\mathrm{Agent~}i\mathrm{~gets~no~bundle~}\\
\Rightarrow~~~&u_i(B)-\sum_{j\in G}B_j\overline{p}_j\leq \delta_w M_i~~~\forall B\in \mathcal{B}
\end{align*}

By choosing $\delta_w$ to be small enough, we can make sure that $2\delta_w M_i\leq \epsilon_u$ for all $i\in N$.
Then we know that given the bundled price rule
\begin{align*}
\mathrm{Price}(B)=\sum_{j\in G}B_j\overline{p}_j
\end{align*}
the agents would choose bundles that yields the best (within acceptable error) payoff, resulting in the designed allocation $\overline{x}$. Thus, prices that support $\overline{x}$ exists.\qed
\section{Proof of Theorem 2.7}
\label{app: envyfreeproof}
Let $\alpha^*_i(i\in N)$ denote some Lagrange multipliers associated with demand constraints in (LIP) (such that $(\alpha^*,p^*)$ is an optimal solution for the dual of (LIP)), from complimentary slackness we know that
\begin{align*}
\alpha^*_i &= \max_{B\in\mathcal{B}}\left\{w_i(B)u_i(B)-\sum_{j\in G}B_j p_j^*, 0\right\}\\
x^*_i(B)>0&\Rightarrow \alpha_i = w_i(B)u_i(B)-\sum_{j\in G}B_j p_j^*\\&\Rightarrow w_i(B)u_i(B)-\sum_{j\in G}B_jp^*_j= \max_{B\in\mathcal{B}}\left\{w_i(B)u_i(B)-\sum_{j\in G}B_j p_j^*\right\}
\end{align*}
\begin{align*}
\sum_{i\in N}x^*_i(B)<1&\Rightarrow \alpha_i = 0\\&\Rightarrow w_i(B)u_i(B)-\sum_{j\in G}B_j p_j^*\leq 0~\forall B\in\mathcal{B}
\end{align*}

Define $M_i=\max_{B\in\mathcal{B}}\{u_i(B)\}$. Consider two cases
\begin{itemize}
	\item $\displaystyle\sum_{i\in N}x_i^*(B)=1$
	
	In this case we have
	\begin{align*}
	\sum_{B\in\mathcal{B}} [w_i(B)u_i(B)-P(B)]x_i^*(B)&=\sum_{B\in\mathcal{B}}\max_{B\in\mathcal{B}}\left\{w_i(B)u_i(B)-\sum_{j\in G}B_j p_j^*\right\}x_i^*(B)\\
	&=\max_{B\in\mathcal{B}}\left\{w_i(B)u_i(B)-P(B)\right\}\cdot \sum_{B\in\mathcal{B}}x_i^*(B)\\
	&=\max_{B\in\mathcal{B}}\left\{w_i(B)u_i(B)-P(B)\right\}\cdot 1\\
	&\geq \max_{B\in\mathcal{B}}\left\{w_i(B)u_i(B)-P(B)\right\}\cdot\sum_{B\in\mathcal{B}}x_k^*(B)\\
	&=\sum_{B\in\mathcal{B}}\max_{B\in\mathcal{B}}\left\{w_i(B)u_i(B)-P(B)\right\}x_k^*(B)\\
	&\geq \sum_{B\in\mathcal{B}} [w_i(B)u_i(B)-P(B)]x_k^*(B)
	\end{align*}
	for all $i,k\in N$.
	
	As $w_i\geq 1-\delta_w$, we then have
	\begin{align*}
	\sum_{B\in\mathcal{B}}[u_i(B)-P(B)]x_i^*(B)&\geq \sum_{B\in\mathcal{B}} [u_i(B)-P(B)]x_k^*(B)\\&-\sum_{B\in\mathcal{B}}(1-w_i(B))u_i(B)[x^*_i(B)-x^*_k(B)]\\
	&\geq \sum_{B\in\mathcal{B}} [u_i(B)-P(B)]x_k^*(B)-\sum_{B\in\mathcal{B}}\delta_w M_i[x^*_i(B)-x^*_k(B)]\\
	&\geq \sum_{B\in\mathcal{B}} [u_i(B)-P(B)]x_k^*(B)-\delta_w M_i
	\end{align*}
	for all $i,k\in N$, and by choosing $\delta_w$ small enough we can ensure that $\delta_wM_i\leq \epsilon_u$ for all $i\in N$.
	\item $\displaystyle\sum_{i\in N}x_i^*(B)<1$
	
	In this case we have $x_i^*(B)>0\Rightarrow w_iu_i(B)-P(B)=0$ 
	
	and $\sum_{i\in N}x_i^*(B)<1\Rightarrow w_i(B)u_i(B)-P(B)\leq 0~\forall B\in\mathcal{B}$. Thus
	\begin{align*}
	\sum_{B\in\mathcal{B}}[w_i(B)u_i(B)-P(B)]x_i^*(B)&=0\geq \sum_{B\in\mathcal{B}}[w_i(B)u_i(B)-P(B)]x_k^*(B)
	\end{align*}
	for all $i,k\in N$.
	
	As $w_i\geq 1-\delta_w$ we have
	\begin{align*}
	\sum_{B\in\mathcal{B}}[u_i(B)-P(B)]x_i^*(B)&\geq \sum_{B\in\mathcal{B}} [u_i(B)-P(B)]x_k^*(B)\\&-\sum_{B\in\mathcal{B}}(1-w_i(B))u_i(B)[x^*_i(B)-x^*_k(B)]\\
	&\geq \sum_{B\in\mathcal{B}} [u_i(B)-P(B)]x_k^*(B)-\delta_w M_i
	\end{align*}
\end{itemize}
for all $i,k\in N$, and by choosing $\delta_w$ small enough we can ensure that $\delta_wM_i\leq \epsilon_u$ for all $i\in N$.

In all, we have the approximate envy-freeness condition holds.\qed

%% file: appendixB.tex
\section{Proof of Polynomial Time}
\label{app: B}

To prove that the algorithm is polynomial time, we first calculate the number of variables in (LIP).

Fix $k$, the number of possible bundles in set $\mathcal{B}$ is given by
\begin{align*}
|\mathcal{B}|=\sum_{m=1}^k \begin{pmatrix}
|G|+m-1\\m
\end{pmatrix}\sim O(|G|^k)
\end{align*}

Therefore the number of variables in (LIP) is given by
\begin{align*}
\mathrm{Nvar} = |N|\cdot |\mathcal{B}|\sim O(|N|\cdot|G|^k)
\end{align*}

The number of constraints in (LIP) is $|N|+|G|$. Since linear programming problem can be solved (average case) in polynomial time with regard to number of variables and number of constraints, the first step of the algorithm can be finished in polynomial time.

For the iterative rounding process, in each iteration, either at least one variable is eliminated or at least one constraint is eliminated. So the number of rounds can not be larger than sum of the number of constraints and the number of variables.

In each round, a linear programming is performed, and it can be solved in polynomial time since its number of constraints and variables is polynomial. Therefore we conclude that the iterative rounding process is also polynomial time.

As the authors of \cite{nguyen2015assignment} proved that the Lottery Construction procedure can finish in polynomial time, we conclude that the whole algorithm can be solved in polynomial time with regard to $|N|$ and $|G|$.\qed

%% file: appendixC.tex
\section{Proof of Theorem 3.2}
\label{app: C}

Given the types of the agents, (LIP) is transformed to
\begin{align*}
\begin{array}{ccll}
\mathrm{(LIP_\Theta)}&\displaystyle\max_{x\geq 0}&\displaystyle\sum_{\theta\in \Theta}\sum_{B\in\mathcal{B}}w_\theta(B) u_\theta(B)y_\theta(B)&\\
&\mathrm{s.t.}&\displaystyle\sum_{B\in\mathcal{B}}\dfrac{1}{n_\theta}y_\theta(B)\leq 1&\theta\in \Theta\\
&&\displaystyle\sum_{\theta\in \Theta}\sum_{B\in \mathcal{B}}B_jy_\theta(B)\leq \tilde{s}_j&j\in G
\end{array}
\end{align*}

After figuring out an optimal solution $x^*$ for (LIP$_\Theta$), individual shares are given by $x^*_i(B)=y^*_{\beta_i}(B)/n_{\beta_i}$. In this sense (LIP$_\Theta$) is equivalent to (LIP) since the symmetrical variables (shares of bundles of agents with same types) in (LIP) can have symmetrical optimal solutions. Therefore we will consider only the solutions of (LIP$_\Theta$) in the following part.

The dual of (LIP$_\Theta$) is given by
\begin{align*}
\begin{array}{ccll}
\mathrm{(DLIP_\Theta)}&\displaystyle\min_{\alpha,p\geq 0}&\displaystyle\sum_{\theta\in \Theta}\alpha_\theta+\displaystyle\sum_{j\in G}\tilde{s}_j p_j\\
&\mathrm{s.t.}&\dfrac{1}{n_\theta}\alpha_\theta\geq w_\theta(B) u_\theta(B)-\displaystyle\sum_{j\in G}B_j p_j&\theta\in \Theta,B\in\mathcal{B}\\
%&&\alpha_\theta\leq \displaystyle\max_{B\in\mathcal{B}}\{w_\theta(B)u_\theta(B)\}&\theta\in \Theta
&&\alpha_\theta\leq M&\theta\in \Theta
\end{array}
\end{align*}
where $M$ is the optimal function value for (LIP$_\infty$), which we will define later.

Note that the last constraint in the above problem is an additional constraint to ensure that the feasible region for the problem is bounded. Adding or deleting this constraint would not affect the optimal solution for the problem, as $M$ is greater or equal to the optimal solution of all (LIP$_\Theta$), and there is no duality gap for linear programming problems, we have $M\geq \sum_{\theta\in\Theta}\alpha_\theta+\sum_{j\in G}s_j'p_j\geq \alpha_\theta$ holds for $\theta\in\Theta$ even if we delete the third constraint. Thus this additional constraint is actually always true.

Consider an agent with type $\xi$ who misreports his type as $\zeta$ while others truthfully report their types, then we have the number of agent reporting each type to be
\begin{align*}
n_{\xi}'&=n_{\xi}-1,~~~n_{\zeta}'=n_{\zeta}+1,\\
n_{\theta}'&=n_{\theta}~~~\theta\in\Theta\backslash\{\xi,\zeta\} 
\end{align*}

Then (LIP$_\Theta'$) in this case is transformed to
\begin{align*}
\begin{array}{ccll}
\mathrm{(LIP_\Theta')}&\displaystyle\max_{x\geq 0}&\displaystyle\sum_{\theta\in \Theta}\sum_{B\in\mathcal{B}}w_\theta(B) u_\theta(B)y_\theta(B)&\\
&\mathrm{s.t.}&\displaystyle\sum_{B\in\mathcal{B}}\dfrac{1}{n_\theta'}y_{\theta}(B)\leq 1&\theta\in \Theta\\
&&\displaystyle\sum_{\theta\in \Theta}\sum_{B\in \mathcal{B}}B_jy_\theta(B)\leq \tilde{s}_j&j\in G
\end{array}
\end{align*}

And the dual of (LIP$_\Theta'$) is given by
\begin{align*}
\begin{array}{ccll}
\mathrm{(DLIP_\Theta')}&\displaystyle\min_{\alpha,p\geq 0}&\displaystyle\sum_{\theta\in \Theta}\alpha_\theta+\displaystyle\sum_{j\in G}\tilde{s}_j p_j\\
&\mathrm{s.t.}&\dfrac{1}{n_\theta'}\alpha_\theta\geq w_\theta(B) u_\theta(B)-\displaystyle\sum_{j\in G}B_j p_j&\theta\in \Theta,B\in\mathcal{B}\\
%&&\dfrac{1}{n_\theta'}\alpha_\theta\leq \displaystyle\max_{B\in\mathcal{B}}\{w_\theta(B)u_\theta(B)\}&\theta\in \Theta
&&\alpha_\theta\leq M&\theta\in \Theta
\end{array}
\end{align*}

When $n_\theta\rightarrow\infty$ for all $\theta\in\Theta$, we have $n_\theta'\rightarrow\infty$ for all $\theta\in\Theta$ and both (LIP$_\Theta$) and (LIP$_\Theta'$) converges to the following problem
\begin{align*}
\begin{array}{ccll}
\mathrm{(LIP_\infty)}&\displaystyle\max_{x\geq 0}&\displaystyle\sum_{\theta\in \Theta}\sum_{B\in\mathcal{B}}w_\theta(B) u_\theta(B)y_\theta(B)&\\
% &&\displaystyle\sum_{B\in\mathcal{B}}\dfrac{1}{n_\theta}y_i(B)\leq 1&i\in N\\
&\mathrm{s.t.}&\displaystyle\sum_{\theta\in \Theta}\sum_{B\in \mathcal{B}}B_jy_\theta(B)\leq \tilde{s}_j&j\in G
\end{array}
\end{align*}

And both (DLIP$_\Theta$) and (DLIP$_\Theta'$) converge to
\begin{align*}
\begin{array}{ccll}
\mathrm{(DLIP_\infty)}&\displaystyle\min_{\alpha,p\geq 0}&\displaystyle\sum_{\theta\in\Theta}\alpha_\theta+\displaystyle\sum_{j\in G}\tilde{s}_j p_j\\
&\mathrm{s.t.}&0\cdot \alpha_\theta\geq w_\theta(B) u_\theta(B)-\displaystyle\sum_{j\in G}B_j p_j&\theta\in \Theta,B\in\mathcal{B}\\
%&&0\cdot\alpha_\theta\leq \displaystyle\max_{B\in\mathcal{B}}\{w_\theta(B)u_\theta(B)\}&\theta\in \Theta
&&\alpha_\theta\leq M&\theta\in \Theta
\end{array}
\end{align*}
which is equivalent to the dual of (LIP$_\infty$), i.e.
\begin{align*}
\begin{array}{ccll}
\mathrm{(DLIP_\infty)}&\displaystyle\min_{\alpha,p\geq 0}&\displaystyle\sum_{j\in G}\tilde{s}_j p_j\\
&\mathrm{s.t.}&0\geq w_\theta(B) u_\theta(B)-\displaystyle\sum_{j\in G}B_j p_j&\theta\in \Theta,B\in\mathcal{B}
\end{array}
\end{align*}
\iffalse 

Note that $\alpha_i$ is reduced in (DLIP$_\infty$). One way to think how (DLIP$_\Theta$) or (DLIP$_\Theta'$) (with $\alpha, p$ as variables) converges to (DLIP$_\infty$) (with only $p$ as variables) is to consider (DLIP$_\Theta$) or (DLIP$_\Theta'$) as non-linear problems of $p$ (but not of $\alpha$), since we can see that $\alpha_\theta=n_\theta\max_{B\in\mathcal{B}}(w_\theta(B)u_\theta(B)-\sum_{j\in G}B_jp_j)$ (or $\alpha_\theta=n_\theta'\max_{B\in\mathcal{B}}(w_\theta(B)u_\theta(B)-\sum_{j\in G}B_jp_j)$) for $\theta\in \Theta$ for any solution of (DLIP$_\Theta$) or (DLIP$_\Theta'$).\fi

We have the following lemmas.\\

\begin{lemma}
\label{lemma:C.1}
With probability 1, (LIP$_\infty$) and (DLIP$_\infty$) both have a unique solution.
\end{lemma}

\begin{proof}
See Appendix \ref{app: proof1}.
\end{proof}

\begin{lemma}
\label{lemma:C.2}
For a sequence of optimization problems $P^{(n)}$ that have common continuous objective function, and the feasible regions $R^{(n)}$ are compact, converges, and satisfies the condition that $\bigcup_{n=1}^\infty R^{(n)}$ is bounded. Let $P*$ be an optimization problem with the same objective function as all $P^{(n)}$ and feasible region $R^*:=\displaystyle\lim_{n\rightarrow\infty} R^{(n)}$. If $P^*$ has a \emph{unique} optimal solution $y^*$, then any sequence $y^{(n)}$ of solutions of $P^{(n)}$ must converge to $y^*$.
\end{lemma}

\begin{proof}
See Appendix \ref{app: proof2}.
\end{proof}

Let $\mathbf{n}^t=(n_{\theta_1}^t, n_{\theta_2}^t, \cdots, n_{\theta_m}^t)$ be any sequence of vectors such that $\mathbf{n}^t\rightarrow\infty$ as $t\rightarrow\infty$. Denote the corresponding sequence of problems as (LIP$_\Theta$)$^t$, (LIP$_\Theta'$)$^t$, (DLIP$_\Theta$)$^t$, and (DLIP$_\Theta'$)$^t$ respectively. And solution sequences for the problem sequences are $y_\Theta^t$, $y_\Theta'^t$, $p_\Theta^t$, $p_\Theta'^t$ respectively. 

We know that (LIP$_\infty$) has a unique solution with probability 1. Since the union of feasible regions of (LIP$_\Theta$)$^t$ and (LIP$_\Theta'$)$^t$ is a subset of feasible region of (LIP$_\infty$), which is bounded, and the objective function are the same, we can apply Lemma C.2 and conclude that with probability 1, $y_\Theta^t$ and $y_\Theta'^t$ both converge to $y_\Theta^*$, a solution of (LIP$_\infty$).

Similarly, we know that (DLIP$_\infty$) has a unique solution with probability 1. Since the union of feasible regions of (DLIP$_\Theta$)$^t$ and (DLIP$_\Theta'$)$^t$ is a subset of feasible region of (DLIP$_\Theta~|~n_\theta=1\forall\theta\in\Theta$), which is bounded (with the additional constraint), and the objective functions are the same, we can apply Lemma C.2 and conclude that with probability 1, $p_\Theta^t$ and $p_\Theta'^t$ both converge to $p_\Theta^*$, a solution of (DLIP$_\infty$).

Then, apply Theorem \ref{thm: envyfree} to (LIP$_\Theta'$)$^t$ we have
\begin{align*}
&\dfrac{1}{n_{\xi}'^t}\sum_{B\in\mathcal{B}}\left[u_\xi(B)-\sum_{j\in G}B_jp^*_j\right]y^{**(t)}_\xi(B)\geq \dfrac{1}{n_{\zeta}'^t}\sum_{B\in\mathcal{B}}\left[u_\xi(B)-\sum_{j\in G}B_jp^*_j\right]y^{**(t)}_\zeta(B)-\epsilon_u\tag{I}
\end{align*}

where $y^{**(t)}$ is an optimal solution of (LIP$_\Theta'$)$^t$. Denote $y^{*(t)}$ as a solution for the corresponding original problem (LIP$_\Theta$)$^t$. We have $y^{*(t)}$ and $y^{**(t)}$ converge to the same limit as $n_\theta\rightarrow\infty$ (with probability 1). Thus we have $\|y^{*(t)}-y^{**(t)}\|\rightarrow 0$ as $n_\theta\rightarrow\infty$ with probability 1. Denote the corresponding price vector as $p^{*(t)}$ and $p^{**(t)}$. As $p^{*(t)}$ and $p^{**(t)}$ converge to the same limit, we also have $\|p^{*(t)}-p^{**(t)}\|\rightarrow 0$ as $n_\theta\rightarrow\infty$ with probability 1. Therefore we have for any $\varepsilon_0>0$ there exist $N_0$ such that when $n_\theta^t>N_0$ for all $\theta\in\Theta$ we have
\begin{align*}
&\dfrac{1}{n_{\xi}^t}\sum_{B\in\mathcal{B}}\left[u_\xi(B)-\sum_{j\in G}B_jp^{*(t)}_j\right]y^{*(t)}_\xi(B)=\dfrac{1}{n_{\xi}'^t+1}\sum_{B\in\mathcal{B}}\left[u_\xi(B)-\sum_{j\in G}B_jp^{*(t)}_j\right]y^{*(t)}_\xi(B)\\=~& \dfrac{1}{n_{\xi}'^t}\sum_{B\in\mathcal{B}}\left[u_\xi(B)-\sum_{j\in G}B_jp^{*(t)}_j\right]y^{*(t)}_\xi(B)-\dfrac{1}{n_{\xi}'^t(n_{\xi}'^t+1)}\sum_{B\in\mathcal{B}}\left[u_\xi(B)-\sum_{j\in G}B_jp^{*(t)}_j\right]y^{*(t)}_\xi(B)\\
\geq~&\dfrac{1}{n_{\xi}'^t}\left(\sum_{B\in\mathcal{B}}\left[u_\xi(B)-\sum_{j\in G}B_jp^{**(t)}_j\right]y^{**(t)}_\xi(B)-\dfrac{1}{2}\varepsilon_0\right)-o((n_\xi^t)^{-1})\\
\geq ~&\dfrac{1}{n_{\xi}'^t}\sum_{B\in\mathcal{B}}\left[u_\xi(B)-\sum_{j\in G}B_jp^{**(t)}_j\right]y^{**(t)}_\xi(B)-\varepsilon_0\tag{II}
\end{align*}
with probability 1.

Combining (I)(II) we get for any $\varepsilon_0>0$ there exist $N_0$ such that when $n_\theta>N_0$ for all $\theta\in\Theta$ we have
\begin{align*}
&\dfrac{1}{n_{\xi}}\sum_{B\in\mathcal{B}}\left[u_\xi(B)-\sum_{j\in G}B_jp^*_j\right]y^{*}_\xi(B)\geq \dfrac{1}{n_{\zeta}'}\sum_{B\in\mathcal{B}}\left[u_\xi(B)-\sum_{j\in G}B_jp^{**}_j\right]y^{**}_\zeta(B)-\epsilon_u-\varepsilon_0
\end{align*}
with probability 1. 

Since the final allocation of POPT is implemented as a lottery, where each agent of type $\theta$ has allocation vector with expectation $\frac{1}{n_\theta}y^*_\theta$ for (LIP$_\Theta$) (or $\frac{1}{n_\theta'}y^{**}_\theta$, for (LIP$_\Theta'$)). Denote the final overall allocation vector as $\overline{x}$ for (LIP) and $\tilde{x}$ for (LIP'). Then we have for any $\varepsilon_0>0$ there exist $N_0$ such that when $n_\theta>N_0$ for all $\theta\in\Theta$ such that
\begin{align*}
&\mathbb{E}\left\{\sum_{B\in\mathcal{B}}\left[u_{\beta_i}(B)-\sum_{j\in G}B_jp^*_j\right]\overline{x}_{\beta_i}(B)\right\}\geq \mathbb{E}\left\{\sum_{B\in\mathcal{B}}\left[u_{\beta_i}(B)-\sum_{j\in G}B_jp^{**}_j\right]\tilde{x}_\zeta(B)\right\}-\epsilon_u-\varepsilon_0
\end{align*}

Proving the result.\qed
\iffalse
As we know that both (LIP$_\Theta$) and (LIP$_\Theta'$) converges to (LIP$_\infty$) as $n_\theta\rightarrow\infty$, the optimal objective function value of (LIP$_\Theta$) and (LIP$_\Theta'$) both converges. Since with probability 1, we have (LIP$_\infty$) to have a unique utility for each type for all solutions, the solutions of (LIP$_\Theta$) and (LIP$_\Theta'$) converges to the solution of (LIP$_\infty$) with probability 1.
\fi

%% file: appendixD.tex
\section{Proof of Lemma D.1}
\label{app: proof1}

In (LIP$_\infty$) we have two set of parameters controlled by random variables: $w_i(B)\sim\mathcal{U}(1-\delta_w,1+\delta_w)$, and $\tilde{s}_j\sim\mathcal{U}(s_j-2\delta_\epsilon, s_j-\delta_\epsilon)$. All of the random variables are independent. 

The constraints of (LIP$_\infty$) defines a (bounded) polytope with finite number of extreme points. Denote $\{q^1, q^2,\cdots, q^m\}$ as the set of extreme points, we can see that $q^1,q^2,\cdots,q^m$ are random variables that depend only on $s_j'$ for $j\in G$.

Therefore, $q^1, q^2,\cdots, q^m$ are independent of $w_i(B)$ for $i\in N,B\in\mathcal{B}$. Define $v^1, v^2,\cdots, v^m$ such that $v^t_i(B)=q^t_i(B)u_i(B)$ for $t=1,2,\cdots,m$. As $q_s\neq q_t$ for $s\neq t$ and $u_i(B)>0$, we have $v_s\neq v_t$ for $s\neq t$. Then we can express the objective function value at $q^t$ as $w^Tq^t$.

Then we have for any $s\neq t$, $s,t=1,2,\cdots,m$,
\begin{align*}
\Pr[w^Tq^s=w^Tq^t]=\Pr[w^T(q^s-q^t)=0]
\end{align*}

Since $(q^s-q^t)$ is a non-zero vector which is independent of $w$, $w^T(q^s-q^t)$ is a continuous random variable. Then the probability that this random variable has a specific value is 0. Thus
\begin{align*}
\Pr[w^Tq^s=w^Tq^t]=0
\end{align*}

It follows that
\begin{align*}
\Pr[w^Tq^s\neq w^Tq^t~\forall s\neq t]=1
\end{align*}

A lemma in \cite{ziegler1995lectures} states that a sufficient condition for the uniqueness of a solution for a linear programming problem is that $w^Tq^s\neq w^Tq^t$ for all $t\neq s$, $s,t=1,2,\cdots, m$. Therefore, with probability 1, we have (LIP$_\infty$) to have a unique solution.\qed

Similarly, in (DLIP$_\infty$) we have a set of parameters controlled by random variables: $\tilde{s}_j\sim\mathcal{U}(s_j-2\delta_\epsilon, s_j-\delta_\epsilon)$. All of the random variables are independent. The constraints defines a bounded polytope with finite number of extreme points. The coordinates of the extreme points is a constant (hence independent of $\tilde{s}_j$). Then, follow a similar proof for (DLIP$_\infty$), we conclude that (DLIP$_\infty$) has a unique solution with probability 1.

Therefore, with probability 1, (LIP$_\infty$) and (DLIP$_\infty$) both have a unique solution.\qed

%% file: appendixE.tex
\section{Proof of Lemma D.2}
\label{app: proof2}

Denote the objective function of $P^{(n)}$ and $P^*$ as $f$, and a sequence of solution of $P^{(n)}$ as $y^{(n)}$. Denote the unique solution of $P^*$ as $y^*$

For the sequence of optimization problems $P^{(n)}$ and $P^*$, if $P^*$ has solutions, then by Berge's Maximum Theorem the optimal objective function values of $P^{(n)}$ converges to the optimal objective function value of $P^*$. That is $
f(y^{(n)})\rightarrow f(y^*)
$ as $n\rightarrow\infty$.

Suppose that $y^{(n)}$ does not converge to $y^*$, which means that there exists some $\varepsilon>0$ such that there exists a subsequence $y^{(n_k)}$ satisfies
\begin{align*}
\|y^{(n_k)}-y^*\|>\varepsilon
\end{align*}

However, $y^{(n_k)}$ is bounded. By Bolzano-Weierstrass Theorem, there exist a subsequence $y^{(n_{k_t})}$ of $y^{(n_k)}$ such that $y^{(n_{k_t})}$ converges to some point $\tilde{y}$. By the definition of set convergence \cite{rockafellar2009variational} the feasible regions $R^{(n_{k_t})}$ converges to the feasible region $R^*$ of $P^*$. $$\lim_{t\rightarrow\infty}R^{(n_{k_t})}=\limsup_{t\rightarrow\infty} R^{(n_{k_t})}=R^*$$ By definition of limits of sets \cite{rockafellar2009variational}, we have $\tilde{y}\in\limsup_{t\rightarrow\infty} R^{(n_{k_t})}$ thus $\tilde{y}\in R^*$. Since the objective function is continuous, we have $f(y^{(n_{k_t})})$ converges to $f(\tilde{y})$.

Since $y^{(n_{k_t})}$ is a subsequence of $y^{(n)}$ and $f(y^{(n)})$ converges to $f(y^*)$, we have $f(y^{(n_{k_t})})$ converges to $f(y^*)$. Therefore $f(\tilde{y})=f(y^*)$, which means that $\tilde{y}$ is also an optimal solution of $P^*$. 

However $y^*\neq \tilde{y}$, as $\|\tilde{y}-y^*\|=\displaystyle\lim_{t\rightarrow\infty}\|y^{(n_{k_t})}-y^*\|\geq \varepsilon>0$. Thus $\tilde{y}$ is an optimal solution different from $y^*$. This contradicts with the condition that $y^*$ is the unique solution for $P^*$. 

Therefore $y^{(n)}$ must converge to $y^*$.\qed

%% file: paper.bbl
\begin{thebibliography}{00}

%%% the Web address.

\ifx \showCODEN    \undefined \def \showCODEN     #1{\unskip}     \fi
\ifx \showDOI      \undefined \def \showDOI       #1{{\tt DOI:}\penalty0{#1}\ }
  \fi
\ifx \showISBNx    \undefined \def \showISBNx     #1{\unskip}     \fi
\ifx \showISBNxiii \undefined \def \showISBNxiii  #1{\unskip}     \fi
\ifx \showISSN     \undefined \def \showISSN      #1{\unskip}     \fi
\ifx \showLCCN     \undefined \def \showLCCN      #1{\unskip}     \fi
\ifx \shownote     \undefined \def \shownote      #1{#1}          \fi
\ifx \showarticletitle \undefined \def \showarticletitle #1{#1}   \fi
\ifx \showURL      \undefined \def \showURL       #1{#1}          \fi

\bibitem[\protect\citeauthoryear{Babaioff, Lucier, Nisan, and
  Paes~Leme}{Babaioff et~al\mbox{.}}{2014}]%
        {babaioff2014efficiency}
{Moshe Babaioff}, {Brendan Lucier}, {Noam Nisan}, {and} {Renato Paes~Leme}.
  2014.
\newblock \showarticletitle{On the efficiency of the walrasian mechanism}. In
  {\em Proceedings of the fifteenth ACM conference on Economics and
  computation}. ACM, 783--800.
\newblock


\bibitem[\protect\citeauthoryear{Bartal, Gonen, and Nisan}{Bartal
  et~al\mbox{.}}{2003}]%
        {bartal2003incentive}
{Yair Bartal}, {Rica Gonen}, {and} {Noam Nisan}. 2003.
\newblock \showarticletitle{Incentive compatible multi unit combinatorial
  auctions}. In {\em Proceedings of the 9th conference on Theoretical aspects
  of rationality and knowledge}. ACM, 72--87.
\newblock


\bibitem[\protect\citeauthoryear{Berry, Honig, and Vohra}{Berry
  et~al\mbox{.}}{2010}]%
        {berry2010spectrum}
{Randall Berry}, {Michael~L Honig}, {and} {Rakesh Vohra}. 2010.
\newblock \showarticletitle{Spectrum markets: motivation, challenges, and
  implications}.
\newblock {\em Communications Magazine, IEEE\/} {48}, 11 (2010), 146--155.
\newblock


\bibitem[\protect\citeauthoryear{Blumrosen and Nisan}{Blumrosen and
  Nisan}{2007}]%
        {blumrosen2007combinatorial}
{Liad Blumrosen} {and} {Noam Nisan}. 2007.
\newblock \showarticletitle{Combinatorial auctions}.
\newblock {\em Algorithmic game theory\/}  {267} (2007), 300.
\newblock


\bibitem[\protect\citeauthoryear{Bulow, Levin, and Milgrom}{Bulow
  et~al\mbox{.}}{2009}]%
        {bulow2009winning}
{Jeremy Bulow}, {Jonathan Levin}, {and} {Paul Milgrom}. 2009.
\newblock {\em Winning play in spectrum auctions}.
\newblock {T}echnical {R}eport. National Bureau of Economic Research.
\newblock


\bibitem[\protect\citeauthoryear{Cramton}{Cramton}{1997}]%
        {cramton1997fcc}
{Peter Cramton}. 1997.
\newblock \showarticletitle{The FCC spectrum auctions: An early assessment}.
\newblock {\em Journal of Economics \&amp; Management Strategy\/} {6}, 3
  (1997), 431--495.
\newblock


\bibitem[\protect\citeauthoryear{Cramton}{Cramton}{2002}]%
        {cramton2002spectrum}
{Peter Cramton}. 2002.
\newblock \showarticletitle{Spectrum auctions}.
\newblock  (2002).
\newblock


\bibitem[\protect\citeauthoryear{Cramton, Skrzypacz, and Wilson}{Cramton
  et~al\mbox{.}}{2007}]%
        {cramton2007700}
{Peter Cramton}, {Andrzej Skrzypacz}, {and} {Robert Wilson}. 2007.
\newblock \showarticletitle{The 700 MHz spectrum auction: An opportunity to
  protect competition in a consolidating industry}.
\newblock  (2007).
\newblock


\bibitem[\protect\citeauthoryear{Cramton, Shoham, Steinberg,
  et~al\mbox{.}}{Cramton et~al\mbox{.}}{2006}]%
        {cramton2006combinatorial}
{Peter~C Cramton}, {Yoav Shoham}, {Richard Steinberg}, {and} {others}. 2006.
\newblock {\em Combinatorial auctions}. Vol. 475.
\newblock MIT press Cambridge.
\newblock


\bibitem[\protect\citeauthoryear{Diamond, Chu, and Boyd}{Diamond
  et~al\mbox{.}}{2014}]%
        {cvxpy}
{Steven Diamond}, {Eric Chu}, {and} {Stephen Boyd}. 2014.
\newblock {CVXPY}: A {P}ython-Embedded Modeling Language for Convex
  Optimization, version 0.2.
\newblock \url{http://cvxpy.org/}.   (May 2014).
\newblock


\bibitem[\protect\citeauthoryear{Dobzinski and Nisan}{Dobzinski and
  Nisan}{2015}]%
        {dobzinski2015multi}
{Shahar Dobzinski} {and} {Noam Nisan}. 2015.
\newblock \showarticletitle{Multi-unit auctions: beyond Roberts}.
\newblock {\em Journal of Economic Theory\/}  {156} (2015), 14--44.
\newblock


\bibitem[\protect\citeauthoryear{Dobzinski, Nisan, and Schapira}{Dobzinski
  et~al\mbox{.}}{2012}]%
        {dobzinski2012truthful}
{Shahar Dobzinski}, {Noam Nisan}, {and} {Michael Schapira}. 2012.
\newblock \showarticletitle{Truthful randomized mechanisms for combinatorial
  auctions}.
\newblock {\it J. Comput. System Sci.} {78}, 1 (2012), 15--25.
\newblock


\bibitem[\protect\citeauthoryear{Domahidi, Chu, and Boyd}{Domahidi
  et~al\mbox{.}}{2013}]%
        {ecos}
{Alexander Domahidi}, {Eric Chu}, {and} {Stephen Boyd}. 2013.
\newblock \showarticletitle{ECOS: An SOCP solver for embedded systems}. In {\em
  Control Conference (ECC), 2013 European}. IEEE, 3071--3076.
\newblock


\bibitem[\protect\citeauthoryear{Easley and Kleinberg}{Easley and
  Kleinberg}{2010}]%
        {easley2010networks}
{David Easley} {and} {Jon Kleinberg}. 2010.
\newblock {\em Networks, crowds, and markets: Reasoning about a highly
  connected world}.
\newblock Cambridge University Press.
\newblock


\bibitem[\protect\citeauthoryear{Guruswami, Hartline, Karlin, Kempe, Kenyon,
  and McSherry}{Guruswami et~al\mbox{.}}{2005}]%
        {guruswami2005profit}
{Venkatesan Guruswami}, {Jason~D Hartline}, {Anna~R Karlin}, {David Kempe},
  {Claire Kenyon}, {and} {Frank McSherry}. 2005.
\newblock \showarticletitle{On profit-maximizing envy-free pricing}. In {\em
  Proceedings of the sixteenth annual ACM-SIAM symposium on Discrete
  algorithms}. Society for Industrial and Applied Mathematics, 1164--1173.
\newblock


\bibitem[\protect\citeauthoryear{Krishna}{Krishna}{2009}]%
        {krishna2009auction}
{Vijay Krishna}. 2009.
\newblock {\em Auction theory}.
\newblock Academic press.
\newblock


\bibitem[\protect\citeauthoryear{Milgrom}{Milgrom}{1998}]%
        {milgrom1998game}
{Paul Milgrom}. 1998.
\newblock \showarticletitle{Game theory and the spectrum auctions}.
\newblock {\em European Economic Review\/} {42}, 3 (1998), 771--778.
\newblock


\bibitem[\protect\citeauthoryear{Mu'Alem and Nisan}{Mu'Alem and Nisan}{2008}]%
        {mu2008truthful}
{Ahuva Mu'Alem} {and} {Noam Nisan}. 2008.
\newblock \showarticletitle{Truthful approximation mechanisms for restricted
  combinatorial auctions}.
\newblock {\em Games and Economic Behavior\/} {64}, 2 (2008), 612--631.
\newblock


\bibitem[\protect\citeauthoryear{Nguyen, Peivandi, and Vohra}{Nguyen
  et~al\mbox{.}}{2015}]%
        {nguyen2015assignment}
{Th\`{a}nh Nguyen}, {Ahmad Peivandi}, {and} {Rakesh Vohra}. 2015.
\newblock \showarticletitle{Assignment Problems with Complementarities}.
\newblock  (2015).
\newblock


\bibitem[\protect\citeauthoryear{Nguyen and Vohra}{Nguyen and Vohra}{2014}]%
        {nguyen2014near}
{Thanh Nguyen} {and} {Rakesh Vohra}. 2014.
\newblock \showarticletitle{Near Feasible Stable Matchings with
  Complementarities}.
\newblock {\em working paper\/} (2014).
\newblock


\bibitem[\protect\citeauthoryear{Nisan}{Nisan}{2000}]%
        {nisan2000bidding}
{Noam Nisan}. 2000.
\newblock \showarticletitle{Bidding and allocation in combinatorial auctions}.
  In {\em Proceedings of the 2nd ACM conference on Electronic commerce}. ACM,
  1--12.
\newblock


\bibitem[\protect\citeauthoryear{Nisan and Ronen}{Nisan and Ronen}{2007}]%
        {nisan2007computationally}
{Noam Nisan} {and} {Amir Ronen}. 2007.
\newblock \showarticletitle{Computationally Feasible VCG Mechanisms.}
\newblock {\em J. Artif. Intell. Res.(JAIR)\/}  {29} (2007), 19--47.
\newblock


\bibitem[\protect\citeauthoryear{Nisan, Roughgarden, Tardos, and
  Vazirani}{Nisan et~al\mbox{.}}{2007}]%
        {nisan2007algorithmic}
{Noam Nisan}, {Tim Roughgarden}, {Eva Tardos}, {and} {Vijay~V Vazirani}. 2007.
\newblock {\em Algorithmic game theory}. Vol.~1.
\newblock Cambridge University Press Cambridge.
\newblock


\bibitem[\protect\citeauthoryear{Rockafellar and Wets}{Rockafellar and
  Wets}{2009}]%
        {rockafellar2009variational}
{R~Tyrrell Rockafellar} {and} {Roger J-B Wets}. 2009.
\newblock {\em Variational analysis}. Vol. 317.
\newblock Springer Science \& Business Media.
\newblock


\bibitem[\protect\citeauthoryear{Shapiro, Holtz-Eakin, and Bazelon}{Shapiro
  et~al\mbox{.}}{2014}]%
        {shapiro2014economic}
{Robert~J Shapiro}, {Douglas Holtz-Eakin}, {and} {Coleman Bazelon}. 2014.
\newblock \showarticletitle{The economic implications of restricting spectrum
  purchases in the incentive auctions}.
\newblock {\em Available at SSRN\/} (2014).
\newblock


\bibitem[\protect\citeauthoryear{Vohra and Krishnamurthi}{Vohra and
  Krishnamurthi}{2012}]%
        {vohra2012principles}
{Rakesh Vohra} {and} {Lakshman Krishnamurthi}. 2012.
\newblock {\em Principles of Pricing: an analytical approach}.
\newblock Cambridge University Press.
\newblock


\bibitem[\protect\citeauthoryear{Zhou, Berry, Honig, and Vohra}{Zhou
  et~al\mbox{.}}{2013}]%
        {zhou2013complexity}
{Hang Zhou}, {Randall Berry}, {Michael~L Honig}, {and} {Rakesh Vohra}. 2013.
\newblock \showarticletitle{Complexity of allocation problems in spectrum
  markets with interference complementarities}.
\newblock {\em Selected Areas in Communications, IEEE Journal on\/} {31}, 3
  (2013), 489--499.
\newblock


\bibitem[\protect\citeauthoryear{Ziegler}{Ziegler}{1995}]%
        {ziegler1995lectures}
{G{\"u}nter~M Ziegler}. 1995.
\newblock {\em Lectures on polytopes}. Vol. 152.
\newblock Springer Science \& Business Media.
\newblock


\bibitem[\protect\citeauthoryear{Zurel and Nisan}{Zurel and Nisan}{2001}]%
        {zurel2001efficient}
{Edo Zurel} {and} {Noam Nisan}. 2001.
\newblock \showarticletitle{An efficient approximate allocation algorithm for
  combinatorial auctions}. In {\em Proceedings of the 3rd ACM conference on
  Electronic Commerce}. ACM, 125--136.
\newblock


\end{thebibliography}
